  \documentclass[11pt,reqno]{amsart}

  \usepackage{amsmath,amsfonts,amsthm,amssymb,amsxtra}
  \usepackage{array}
  \usepackage{point}
  \usepackage[hidelinks]{hyperref}

  \newtheorem{theorem}{Theorem}
  
  \newtheorem{lemma}[theorem]{Lemma}

  \newtheorem{remark}[theorem]{Remark}

  \theoremstyle{definition}

  \newtheorem{definition}[theorem]{Definition}

  \numberwithin{equation}{section}
  \numberwithin{theorem}{section}

  \renewcommand{\epsilon}{\varepsilon}

  \newcommand{\G}{{\mathcal G}}

  \renewcommand{\phi}{\varphi}
  \newcommand{\R}{\mathbb{R}}

\begin{document}

  \title[Analysis of a simple equation for the ground state of the Bose gas II]{Analysis of a simple equation for the ground state of the Bose gas II: Monotonicity, Convexity and Condensate Fraction}

  \author{Eric A. Carlen}
  \address[Eric A. Carlen]{Department of Mathematics, Hill Center, Rutgers University, 110 Frelinghuysen Road, Piscataway, NJ 08854-8019, USA}
  \email{carlen@math.rutgers.edu}

  \author{Ian Jauslin}
  \address[Ian Jauslin]{Department of Physics, Princeton University, Jadwin Hall, Washington Road, Princeton, NJ 08544, USA}
  \email{ijauslin@princeton.edu}

  \author{Elliott H. Lieb}
  \address[Elliott H. Lieb]{Departments of Mathematics and Physics, Princeton University, Jadwin Hall, Washington Road, Princeton, NJ 08544, USA}
  \email{lieb@math.princeton.edu}

  \maketitle

\begin{abstract}
In a recent paper we studied an equation (called the ``simple equation'') introduced by one of us in 1963 for an approximate correlation function associated to the ground state of an interacting Bose gas.
Solving the equation yields a relation between the density $\rho$ of the gas and the energy per particle.
Our construction of solutions  gave a well-defined function $\rho(e)$ for the density as a function of the energy $e$.
We had conjectured that $\rho(e)$ is a strictly monotone increasing function, so that it can be inverted to yield the strictly monotone increasing function $e(\rho)$.
We had also conjectured that $\rho e(\rho)$ is convex as a function of $\rho$.
We prove both conjectures here for small densities, the context in which they have the most physical relevance, and the monotonicity also for large densities.
Both conjectures are grounded in the underlying physics, and their proof provides further mathematical evidence for the validity of the assumptions underlying the derivation of the simple equation, at least for low or high densities, if not intermediate densities, although the equation gives surprisingly good predictions for all densities $\rho$.
Another problem left open in our previous paper was whether the simple equation could be used to compute accurate predictions of observables other than the energy.
Here, we provide a recipe for computing predictions for any one- or two-particle observables for the ground state of the Bose gas.
We focus on the condensate fraction and the momentum distribution, and show that they have the same low density asymptotic behavior as that predicted for the Bose gas.
Along with the computation of the low density energy of the simple equation in our previous paper, this shows that the simple equation reproduces the known and conjectured properties of the Bose gas at low densities.
\end{abstract}

  \tableofcontents

  \thanks{\hfil\small\copyright\, 2020 by the authors. This paper may be reproduced, in its entirety, for non-commercial purposes.}

  \section{Introduction}

  We study the system of equations
  \begin{equation}\label{simp}
  (-\Delta + 4e + v(x))u(x) =  v(x) + 2e\rho u*u(x)\ 
  ,\quad
  \frac{2e}{\rho}=\int (1-u(x))v(x)\ dx
  \end{equation}
  to be solved for an integrable function $u$ on $\R^3$ where $v$ is a given non-negative radial function representing a repulsive interaction between particles with $(1+|x|^4)v\in L^1(\R^3)\cap L^2(\R^3)$, and
  where $e$ and $\rho$ are positive parameters representing, respectively, the energy per-particle and the density in the ground state of a Bose gas, and are related by the second equation in\-~(\ref{simp}).  As we explain below, the solution $u(x)$ specifies a pair correlation function for the Bose gas in terms of which many observable of physical interest can be computed. 
  This system was first introduced in \cite{Li63,LS64,LL64} and the equation on the left  is referred to here as the {\em simple equation}; it results from applying some approximations to a more complicated equation derived in \cite{Li63}.
  For the reader interested in the origins of this equation, we give a brief account of its derivation and motivation.
  The simple equation arises in connection with the ground state $\psi_0$ of a many-body Bose gas, whose $N$-particle Hamiltonian is given by
  \begin{equation}\label{ham0}
    H_N:=-\frac12\sum_{i=1}^N\Delta_i+\sum_{i<j}v(x_i-x_j)
  \end{equation}
  for $N$ particles in a cubic box of  finite volume $V$ with periodic boundary conditions.
  The ground state eigenfunction $\psi_0$ is unique and non-negative, as can be shown using the Perron-Frobenius theorem, and thus we may normalize $\psi_0$ to obtain a probability measure.  This is not the usual probability measure associated to a quantum state, which would be quadratic in the wave function, but since $\psi_0$ is non-negative and integrable ($\|\psi_0\|_1 \leqslant V^{1/2}\|\psi_0\|_2$), we may use it directly to define a probability measure, and this is the starting point of \cite{Li63}.  Because particles interact pairwise,  the ground state energy and other observables can be calculated in terms of the two-point correlation function associated to this probability measure:
\begin{equation}
  g(x_1-x_2) := \lim_{N,V\to\infty, N/V = \rho}\frac{V^2\int dx_3\cdots dx_N\ \psi_0(x_1,x_2,x_3,\dots,x_N)}{\int dy_1\cdots dy_N\ \psi_0(y_1,\dots,y_N )}
\end{equation}
  In \cite{Li63}, under a few physically motivated approximations, in the thermodynamic limit, in which the number of particles $N$ and the volume of the gas $V$ are taken to infinity, with $\rho:=\frac NV$ fixed, an equation for the limiting two-point correlation function $g$ is derived.
  The function $u(x)$ in \eqref{simp} is then defined as $u(x):=1- g(x)$. Note that since by definition $g(x) \geqslant 0$, $u(x) \leqslant 1$.

  Because the expected values in the ground state of  many physical observables can be calculated in terms of $g$, any method for computing $g$ that bypasses directly solving the $N$-body Schr\"odinger equation for the Hamiltonian  \eqref{ham0} provides an effective means for the computation of these  values, and this motivates the study of  the simple equation system\-~(\ref{simp}). Indeed, the ground state energy per particle is given in terms of $g$ by the second equation in \eqref{simp}.  There is so far no rigorous derivation of \eqref{simp} from the $N$-body Schr\"odinger equation,  and hence there is no mathematical 
   understanding of how closely the solutions of\-~(\ref{simp}) approximate the actual two point correlation function associated to the $N$-body ground state $\psi_0$. However, we have conducted extensive numerical work, about\-~(\ref{simp}) and other, more refined equations, and have found that these equations are surprisingly accurate.
  Details on the numerical results will be published elsewhere\-~\cite{CHe}.

  The  ground state of the many-body Bose gas in the thermodynamic limit is still the focus of much current research.
  While there are many results in other scaling regimes (see, to name but a few, \cite{LS02,GS09,Se11,BBe18}; for a more comprehensive review, see\-~\cite{LSe05}) rigorous results in the thermodynamic limit are mostly focused on the ground state energy \cite{Bo47,LHY57,Dy57,LY98,ESY08,YY09,FS20}.
  Notably, it was recently shown \cite{YY09,FS20} that for the Bose gas the ground state energy behaves, as $\rho\to0$, as
  \begin{equation}\label{LHY}
    e(\rho)=2\pi\rho a_0\left(1+\frac{128}{15\sqrt\pi}(\rho a_0^3)^{\frac12}+o(\rho^{\frac12})\right)
  \end{equation}
  where $a_0$ is the {\em scattering length} of the potential $v$; see\-~\cite{LSe05}.
  However, a more precise understanding of the physics of the ground state is still lacking.
  In particular, it is expected that the Bose gas should exhibit {\it Bose-Einstein condensation},  in which a macroscopic number of particles occupy the same quantum state.
  So far, Bose-Einstein condensation has only been proved in the thermodynamic limit for a lattice gas in dimensions $\geqslant 2$ at half filling\-~\cite{KLS88}, as well as in other scaling regimes, such as the Gross-Pitaevskii regime \cite{LS02,BBe18}.
  It has never been proved in the thermodynamic limit for a continuum system.

  If it is indeed true that the simple equation  describes the many-body Bose gas in the thermodynamic limit with meaningful accuracy, then it seems important to understand this equation beyond simple numerics.
  We have started this effort in a previous publication\-~\cite{CJL20}, where we showed that, under the assumption that $v\geqslant 0$ and that $v\in L^1(\mathbb R^d)\cap L^{\frac32+\epsilon}(\mathbb R^d)$ but not necessarily radial, then in each dimension $d$,  for each $e>0$, there is a unique value $\rho(e)$ for which \eqref{simp} has an integrable solution satisfying $u \leqslant 1$, and for each $e>0$, there is exactly one integrable solution $u$ with $u \leqslant 1$. (Recall that $u\leqslant 1$ is equivalent to $g\geqslant 0$, a necessary condition for the solution to be physically meaningful.)  We also proved that all such solutions are necessarily non-negative, so that
  \begin{equation}\label{uinf}
    0\leqslant u(x)\leqslant 1
    .
  \end{equation}
  Although the two parameters $e$ and $\rho$ appear to enter \eqref{simp} in a symmetric way, this is not the case in the analysis \cite{CJL20}.  We first fix $e$, and then construct $\rho(e)$ and the corresponding solution $u$ in an iterative process. We show that the function $\rho(e)$ that we construct is continuous, but the analysis in  \cite{CJL20} does not show that $\rho(e)$ is strictly monotone increasing in $e$, which would permit us to invert the functional relationship and define the function $e(\rho)$, which  of course would then also be strictly monotone. In \cite{CJL20}, we showed that for each $\rho > 0$, there was {\em at least} one $e$ such that $\rho = \rho(e)$, and that
  \begin{equation}
  \frac{2e}{\|v\|_1}\leqslant \rho(e) \leqslant  \frac{4e}{\|v\|_1} \ ,
   \label{con4B}
  \end{equation}
  see (1.21) in \cite{CJL20}),
  and finally we showed that for any such $e$, \eqref{LHY} was satisfied, following the lines of a calculation in \cite{Li63}.

  In addition, we showed that, under the further assumption that $v$ is of positive type (its Fourier transform is non-negative), the quantity $e$ defined in\-~(\ref{simp}) coincides with the ground state-energy per particle of the many-body Bose gas, asymptotically both for small and large values of $\rho$.
  Finally, we showed that, if the potential $v$ is spherically symmetric and decays exponentially, then $u\sim|x|^{-4}$ for large $|x|$.

  In the present paper we take this analysis further, and prove some of the conjectures in\-~\cite{CJL20}, namely that the map $\rho\mapsto e(\rho)$ is strictly monotone increasing for small and for large $\rho$, as well as the fact that the map $\rho\mapsto \rho e(\rho)$ is convex for small values of $\rho$.
  Both of these properties hold for the many-body Bose gas: indeed, the monotonicity follows simply from the fact that $v\geqslant 0$, and the convexity statement is equivalent to saying that the compressibility of the Bose gas is positive (that is, if the gas is compressed, then the pressure increases).
  In addition, whereas the analysis in\-~\cite{CJL20} focused solely on the energy of the Bose gas, we will show that the simple equation can be used to compute an approximation for any one-particle observable.
  In particular, we show that the condensate fraction (that is, the proportion of particles in the Bose-Einstein condensate) agrees with the prediction by Bogolyubov \cite{LHY57}.
  This is rather significant since, if we could show that  the simple equation approximates the Bose gas, this would imply the existence of a Bose-Einstein condensate in the thermodynamic limit.
  Furthermore, we show that the $|x|^{-4}$ decay proved in\-~\cite{CJL20} can be extended to a much larger class of potentials.
  Finally, we exhibit an explicit solution of\-~(\ref{simp}) for a special potential.

  These results solve some of the open problems posed in\-~\cite{CJL20}, though others remain unsolved.
  In particular, the monotonicity result only holds for small and large densities, and the convexity result only for small densities.
  We conjecture that this should be true for all densities, but do not have a proof for this.
  Another open problem concerns the so-called {\it full equation} (see\-~\cite[(7.2)]{CJL20}), which is  the other more intricate  effective equation for the two-point correlation function that was mentioned above, and of which the {\it simple equation} is an approximation.
  Though our numerical results \cite{CHe} predict that the full equation is very accurate in reproducing the behavior of the ground state of the many-boson system, there is so far no proof that it admits any solution, let alone theorems about its properties.

  While the results presented in this paper may seem disparate,  for the most part they are obtained through the use of a  common  set of mathematical tools.
  To see this, let us first consider the monotonicity result.
  To prove that the map $e\mapsto\rho(e)$ is monotone increasing,
   formally differentiate\-~(\ref{simp}) with respect to $e$, and find that, denoting derivatives with respect to $e$ by primes,
  \begin{equation}\label{upform}
    u'=\mathfrak K_e(-4u+2\rho u\ast u+2\rho'u\ast u)
  \end{equation}
  with
  \begin{equation}\label{bg6}
    \mathfrak{K}_e  = (-\Delta + v +  4e(1 - C_{\rho u}))^{-1}
  \end{equation}
  in which $C_{\rho u}$ denotes the convolution by $\rho u$.  Now, differentiating the second equation in \eqref{simp} in $e$ yields
  \begin{equation}\label{bg6X}
  \rho' = \frac{\rho}{e} +\frac{\rho^2}{2e}\int u'v\ dx.
  \end{equation}
  Multiplying \eqref{bg6} by $v$ and integrating yields  an expression for $\rho'$ in terms of $e$, $\rho$, $u$ and the operator $\mathfrak{K}_e$:
  \begin{equation}\label{rpratio}
  \frac{e}{\rho} \rho' = \frac { 1 + \rho\int  (   \mathfrak{K}_e v)(\rho u*u - 2u)dx }{1 -\rho^2\int (   \mathfrak{K}_e v) u*u dx   }\ .
  \end{equation}

   Justifying these formal calculations and analyzing the resulting expression for $\rho'$,  we will prove its strict positivity  at all sufficiently low or high densities, and in some cases, depending on $v$, for all densities (see Theorem~\ref{Mon}).   It is easy to see that the same operator $\mathfrak{K}_e$ will again show up in the computations we do to prove convexity of $e\rho(e)$. It is probably less clear that it will again show up when we derive formulas for other observable such as the condensate fraction, and we now explain why this is the case.

  Let $A$ be a self adjoint operator on the $N$-particle Hilbert space, representing some observable whose ground state expectation value $\langle \psi_0,A\psi_0\rangle$ we would like to evaluate. 
  Introduce a real parameter $\mu$ and the perturbed Hamiltonian  
  \begin{equation}
    H_N^{(\mu)}:=
    -\frac12\sum_{i=1}^N\Delta_i+\sum_{i<j}v(x_i-x_j)
    -\mu A
  \end{equation}
  and denote its ground state by $\psi_0^{(\mu)}$ and its energy by $E_N^{(\mu)}$.  Then
  \begin{equation}
    E_N^{(\mu)}=\left<\psi_0\right|H_N^{(\mu)}\left|\psi_0\right>
  \end{equation}
  and
  \begin{equation}\label{FH}
    \left<\psi_0\right|A\left|\psi_0\right>
    =-\partial_\mu E_N^{(\mu)}|_{\mu=0}
    .
  \end{equation}
  The ground state of $-\frac12\Delta$, the kinetic energy for one particle, is the constant function $V^{-1/2}$. Let  $P_j$  denote the projector onto  this state acting on the coordinates of the $j$th particle; i.e., for any $\phi$ in the $N$-particle Hilbert space,
  $$
  P_j\phi(x_1,\dots,x_N) = \int\frac{dx_j}{V}\phi(x_1,\dots x_N)\ .
  $$
  The condensate fraction, denoted by $1-\eta$, is the quantity obtained by taking $A = \frac{1}{N}\sum_{j=1}^N P_j$, and it represents the fraction of the particles in the Bose-Einstein condensate. Thus $\eta$ is the fraction of the particles that are {\em not} in the condensate:
  \begin{equation}\label{cond2}
    1-\eta=
    \frac1N\sum_{j=1}^N\left<\psi_0\right|P_j\left|\psi_0\right>
    .
  \end{equation}
  Following the procedure used in\-~\cite{Li63} to derive the simple equation starting from the Hamiltonian \eqref{ham0}, we start from the perturbed  Hamiltonian $H_N^{(\mu)}$ to  derive a modified simple equation:
  \begin{equation}
    (-\Delta+2\mu+4e_\mu) u_\mu=(1-u_\mu)v+2\rho e_\mu u_\mu\ast u_\mu
    ,\quad
    e_\mu=\frac\rho2\int (1-u_\mu(x))v(x)\ dx
    \label{simpleq_eta}
  \end{equation}
  and then on account of \eqref{FH} we obtain
  \begin{equation}
    \eta=\partial_\mu e_\mu|_{\mu=0}
    .
    \label{eta}
  \end{equation}
  Differentiating \eqref{simpleq_eta} leads once more to the operator $\mathfrak{K}_{e_\mu}$.
Note that, since approximations were made  in computing the two-point correlation function, it is not immediately clear that the quantity $\eta$ defined in \eqref{etaf} satisfies $0 \leqslant \eta \leqslant 1$.
In the rest of this paper, we always use $\eta$ to mean the quantity defined in \eqref{eta}, and not the true uncondensed fraction, defined in \eqref{cond2}.
We shall see that at least for small $\rho$, the approximation is very good. 

\indent
Another observable of interest is the {\it momentum distribution}
\begin{equation}
  \mathfrak M(k):=\frac 1N\sum_{i=1}^N\left<\psi_0\right|K_i(k)\left|\psi_0\right>
\end{equation}
with
\begin{equation}
  K_i(k)\varphi(x_1,\cdots,x_N):=
  \int e^{ik(y_i-x_i)}\varphi(x_1,\cdots,x_{i-1},y_i,x_{i+1},\cdots,x_N)\ dy_i
  \ .
\end{equation}
A well known prediction\-~\cite{CAL09} is that, for a delta-function potential, the momentum distribution should behave asymptotically as $|k|\to\infty$ as\-~\cite[6.2.1.2]{NE17}
\begin{equation}
  \mathfrak M(k)\sim\frac{16\pi^2a^2\rho}{|k|^4}
  \label{tan}
\end{equation}
which is knwon as the {\it universal Tan relation}\-~\cite{Ta08,Ta08b,Ta08c}.
We have found that the simple equation reproduces this prediction, even when the potential is finite, when the density is asymptotically small (see Theorem\-~\ref{theo:tan}).
To compute an approximation for $\mathfrak M(k)$, we follow the same procedure as above, which leads us to the following equation: for $k\neq0$,
\begin{equation}
  (-\Delta+4e_\mu) u_\mu=(1-u_\mu)v+2\rho e_\mu u_\mu\ast u_\mu
  +2\mu\hat u_0(k)\cos(k\cdot x)
  ,\quad
  e_\mu=\frac\rho2\int dx\ (1-u_\mu(x))v(x)
  \label{simpleq_momentum}
\end{equation}
and
\begin{equation}
  \mathfrak M(k)=\frac\rho2\int\  v(x)\partial_\mu u_\mu(x)|_{\mu=0}  dx
  .
  \label{gamma1}
\end{equation}
 
Differentiating \eqref{simpleq_momentum} leads once more to the operator $\mathfrak{K}_{e_\mu}$.

Therefore, a significant part of the analysis in this paper is aimed at understanding the operator $\mathfrak K_e$, as well as properties of solutions $u$ of the simple equation. 
Consider for example the problem of showing that $\rho'(e) > 0$ using the formula in \eqref{rpratio}.  We will need to have $L^p$ to $L^q$ mapping properties of $\mathfrak{K}_e$, among other things, but all $L^p$ bounds on solutions $u$ of the simple equation system.
Integrating both sides of the simple equation, one sees that all solutions of the system satisfy
\begin{equation}\label{intu}
  \int u(x)\ dx=\frac1\rho \ .
\end{equation}
Then since all physical solutions (those satisfying $u(x) \leqslant 1$) satisfy $0 \leqslant u(x) \leqslant 1$, it follows that $u\in L^p(\R^3)$ for all $1 \leqslant p \leqslant \infty$, and the obvious estimate that follows from this information is
$\|u\|_p \leqslant \rho^{-1/p}$.   However,  one can do  significantly better. We shall prove the following lemma in section~\ref{UB}.

\begin{lemma}\label{ul2b}  For $1\leqslant p < 3$, solutions $u$ of \eqref{simp} satisfy
\begin{equation}\label{sim6B}
\|u\|_p \leqslant  C_p e^{-(p-3)/2p}\qquad{\rm where}\qquad C_p := 2(4\pi)^{1/p -1} \Gamma^{1/p}(3-p) (2p)^{(p-3)/p}\|v\|_1 \ .
\end{equation}
In particular,
\begin{equation}\label{sim6}
\|u\|_2 \leqslant \frac{\|v\|_1}{4\sqrt{\pi}}  e^{-1/4}\ ,
\end{equation}
while for large $e$ we have the bound
\begin{equation}\label{sim6Y}
\|u\|_2 \leqslant \frac{1}{2e}\|v\|_2\ .
\end{equation}

\end{lemma}
On account of \eqref{con4B}, this is significantly better than the bound $\|u\|_2 \leqslant \rho^{-1/2}$ that follows trivially from \eqref{intu} and $0\leqslant u(x) \leqslant 1$.  
We shall also need various $L^p$ bounds on $u'$, and for these we need a detailed understanding of the $L^p$ to $L^q$ mapping properties on the operator $\mathfrak{K}_e$. We briefly describe this at the end of the introduction after first describing our main results on the simple equation itself.

\subsection{Main results}
  Our first result on the decay at infinity of solution of the simple equation is used throughout the paper. For example, it is the basis of applications of Lebesgue's Dominated Convergence Theorem to show the formal limit taken in deriving the expression \eqref{upform} do exist. 

\begin{theorem}[Large $|x|$ asymptotics of $u$]\label{theo:pointwise}
  If $(1+ |x|^4)v(x)\in L^1(\mathbb R^3)\cap L^2(\mathbb R^3)$, then
  \begin{equation}
    \rho u(x)=
    \frac{\sqrt{2+\beta}}{2\pi^2\sqrt{e}}\frac1{|x|^{4}}
    +
    R(x)
    \label{udecay}
  \end{equation}
  where
   \begin{equation}\label{betadef}
    \beta  =   \rho \int |x|^2 v(1-u)dx  \leqslant \rho\|x^2v\|_1
    ,
  \end{equation}
  and  where $|x|^4R(x)$ is in $L^2(\mathbb R^3)\cap L^\infty(\R^3)$, uniformly in $e$ on all compact sets.
 Moreover, for every $\rho_0>0$, there is a constant $C$ that only depends on $\rho_0$ such that for all $x$, for all $\rho<\rho_0$,
  \begin{equation}\label{fourdecay}
   u(x) \leqslant \min\left\{1,\frac C{\rho e^{\frac12}|x|^4}\right\}\ .
  \end{equation}
\end{theorem}

The next two theorems concern the monotonicity of $\rho\mapsto e(\rho)$ and convexity if $\rho\mapsto\rho e(\rho)$.
These were conjectured in\-~\cite{CJL20}, and here, we prove them for small density $\rho$ (and, in the case of the monotonicity, also for large density).

\begin{theorem}[Monotonicity]\label{Mon}  Assume that $(1+ |x|^4)v(x)\in L^1(\mathbb R^3)\cap L^2(\mathbb R^3)$.
For 
$$
e < e_\star:=\frac{\sqrt 2 \pi^{3}}{\|v\|_1^2}\qquad{\rm  and\  for}\qquad  e >    \frac{2^3\|v\|_2^4}{\pi^4} $$
 $\rho(e)$ is strictly monotone increasing in $e$, and in these intervals $\rho(e)$ is continuously differentiable.  If $u(e,\cdot)$ denotes the solution of (\ref{simp}) as a function of $e$, $u(e,\cdot)$ is continuously differentiable in $L^2(\R^3)$. Moreover,
\begin{equation}\label{rhopb2}
\mathrm{for}\ e < e_\star\equiv\frac{\sqrt 2 \pi^{3}}{\|v\|_1^2}
\quad\mathrm{we\ have}\quad
\rho'\equiv\frac{d\rho}{de} \leqslant \frac{16}{\|v\|_1}\ .
\end{equation}
\end{theorem}

\begin{remark} Notice that when $\|v\|_2^4\|v\|_1^2 \leqslant 2^{-\frac52}\pi^{7}$, the intervals overlap, and monotonicity holds for all $e$. 
\end{remark}

\begin{theorem}[Convexity]\label{theo:convexity}   Assume that $(1+ |x|^4)v(x)\in L^1(\mathbb R^3)\cap L^2(\mathbb R^3)$ and that $(1+|x|^3)v(x) \in L^8(\R^3)$.
For ${\displaystyle e < e_\star\equiv \frac{\sqrt 2 \pi^{3}}{\|v\|_1^2}}$, $\rho e(\rho)$ is a convex function of $\rho$.  
\end{theorem}

The next theorem concerns the condensate fraction.
In it we provide a formula for the prediction the simple equation makes for the condensate fraction of the many-body Bose gas, and we show that this prediction satisfies the low density asymptotic that is conjectured to hold for the many-body Bose gas.

\begin{theorem}[Condensate fraction]\label{theo:condensate_eta}   Assume that $(1+ |x|^4)v(x)\in L^1(\mathbb R^3)\cap L^2(\mathbb R^3)$.
  The non-condensed fraction $\eta$ defined in\-~(\ref{eta}) satisfies
  \begin{equation}
    \eta=\frac{\rho\int v(x)\mathfrak K_eu(x)\ dx}{1-\rho\int v(x)\mathfrak K_e(2u(x)-\rho u\ast u(x))\ dx}
    .
    \label{etaf}
  \end{equation}
  As $\rho\to0$, $\eta$ goes to 0 asymptotically as
  \begin{equation}
    \eta\sim\frac{8\sqrt{\rho a_0^3}}{3\sqrt\pi}
    \label{eta_asym}
  \end{equation}
  where $a_0$ is the scattering length of $v$.
  This coincides with a well-known prediction for the many-body Bose gas\-~\cite[(41)]{LHY57}.
\end{theorem}

In the next theorem, we show that, in a certain limiting regime, the prediction of the simple equation for the momentum distribution satisfies Tan's universal relation, which are conjectured to hold for the many-body Bose gas.

\begin{theorem}[Momentum distribution]\label{theo:tan}
  Assume that $(1+ |x|^4)v(x)\in L^1(\mathbb R^3)\cap L^2(\mathbb R^3)$.
  The momentum distribution defined in\-~(\ref{gamma1}) satisfies
  \begin{equation}
    \mathfrak M(k)=
    \frac{\rho\hat u_0(k)\int v(x)\mathfrak K_e\cos(k\cdot x)\ dx}
    {1-\rho\int v\mathfrak K_e(2u_0-\rho u_0\ast u_0)\ dx}
    .
    \label{frakM}
  \end{equation}
  Consider the limit $|k|\to0$ and $\rho\to0$ in such a way that $\kappa:=\frac{|k|}{2\sqrt e}\to\infty$.
  In this limit,
  \begin{equation}
    \mathfrak M(k)\sim\frac{1}{4\rho\kappa^4}\sim\frac{C_2}{|k|^4}
    ,\quad
    C_2=:\frac{4e^2}\rho
    \label{M_asym}
  \end{equation}
  which coincides with\-~(\ref{tan}) in the limit $\rho\to0$.
\end{theorem}

Finally, we exhibit an explicit solution to the simple equation in the next theorem.

\begin{theorem}[Explicit solution]\label{theo:explicit}
  For $e,b,c>0$ such that
  \begin{equation}
    \frac e{b^2}\geqslant
    \frac79  \quad{\rm and}\qquad c \leqslant 1\ ,
    \label{cond_explicit}
  \end{equation}
   the function
  \begin{equation}
    u(x)=\frac c{(1+b^2x^2)^2}
  \end{equation}
  is the solution of\-~(\ref{simp}) with $\rho=\frac{b^3}{c\pi^2}$ and the potential
  \begin{equation}
    v(x)=
    \frac{12c( x^6b^6(2e-b^2) +b^4x^4(9e-7b^2) +4b^2x^2(3e-2b^2) +(5e+16b^2))}{(1+b^2x^2)^2(4+b^2x^2)^2((1+b^2x^2)^2-c)}
    \label{v_explicit}
  \end{equation}
  which is in $L^1(\mathbb R^3)\cap L^{\infty}(\mathbb R^3)$, and is non-negative for all $c\leqslant 1$.
 
\end{theorem}

\begin{remark}
  Theorem\-~\ref{theo:explicit} actually holds if the first condition\-~(\ref{cond_explicit}) is replaced by
  \begin{equation}
    \frac e{b^2}\geqslant
    \frac{-263+23\sqrt{161}}{48}
    \approx0.60
  \end{equation}
  which is the necessary and sufficient condition for the numerator in\-~(\ref{v_explicit}) to be non-negative.
  We do not give the proof of this statement here, as it is a bit tedious, and only marginally improves the $\frac79$ constant.
\end{remark}

\subsection{Tools for the proofs: the operator $\mathfrak K_e$ and a variant of the HLS inequality}

We now describe some of the main results on $\mathfrak{K}_e$ that we shall need. 
On account of \eqref{intu},  $\rho u$ is a probability density, and
\begin{equation}
  0 \leqslant 4e(I - C_{\rho u}) \leqslant 4e
\end{equation}
so  that $\mathfrak K_e$, as an operator from $L^1(\mathbb R^3)$ to $L^1(\mathbb R^3)$, is unbounded.  However,  $e^{t4e(I - C_{\rho u})}$ is easily seen to be a positivity preserving contraction semigroup on $L^p$ for all $p$, as is $e^{t(-\Delta + v)}$~\cite{Ne64,RS75b}.  Then by the Trotter Product formula, so is 
$e^{t(-\Delta + v +  4e(1 - C_{\rho u})}$.  Since 
\begin{equation}\label{intrep}
\mathfrak{K}_e  = \int_0^\infty dt e^{t(-\Delta + v +  4e(1 - C_{\rho u})}\ ,
\end{equation}
$\mathfrak{K}_e$ has a positive kernel denoted $\mathfrak{K}_e(x,y)$.   We also define the convolution operator
\begin{equation}\label{sim76}
\mathfrak{Y}_e  := (-\Delta + 4e(1 -  C_{\rho u}))^{-1} \ .
\end{equation}
which is related to $\mathfrak{K}_e$ by the resolvent identity:
\begin{equation}\label{resolid}
\mathfrak{K}_e = \mathfrak{Y}_e - \mathfrak{Y}_e v \mathfrak{K}_e\ .
\end{equation}
Reasoning as above, we conclude that $\mathfrak{Y}_e$ preserves positivity and hence is given by convolution with a non-negative function also denoted $\mathfrak{Y}_e(x)$, and then by \eqref{resolid} and the non-negativity of $v$, 
\begin{equation}\label{KYcomp}
  \mathfrak{K}_e (x,y) \leqslant \mathfrak{Y}_e(x-y)\ .
\end{equation}
The Fourier transform of  $\mathfrak{Y}_e(x)$, $\widehat{\mathfrak{Y}_e}(k)$ is given by 
$$
\widehat{\mathfrak{Y}_e}(k)  =  \left( k^2 + 4e(1 - \rho \widehat{u}(k))     \right)^{-1}\ .
$$
Fourier transforming the simple equation, one finds
\begin{equation}\label{uhat}
  \rho\widehat u(k)=\frac{k^2}{4e}+1-\sqrt{\left(\frac{k^2}{4e}+1\right)^2-\frac\rho{2e}\widehat S(k)}
  ,\quad
  \widehat S(k):=\int dx\ e^{ikx}(1-u(x))v(x)\ .
  \end{equation}
  By \eqref{intu}, $\rho \widehat{u}(0) =1$ and by the second equation in \eqref{simp}, $\frac{\rho}{2e}\widehat{S}(0) =1$, 
 and from here  one obtains 
  \begin{equation}\label{facA}
 \left( k^2 + 4e(1 - \rho \widehat{u}(k))     \right)^{-1}  \leqslant   |k|^{-1}\left(k^2 + 2\sqrt{2e}\right)^{-1/2}  
 \end{equation}
 The right side is square integrable, and in this way we obtain a bound on $\mathfrak{Y}_e$ from $L^1(\R^3)$ to $L^2(\R^3)$. The following lemma (proved in section~\ref{UB}) summarizes information that we obtain on $\mathfrak{K}_e$ that suffices to prove Theorem~\ref{Mon} on monotonicity.

 \begin{lemma}\label{frakKprops}  Let $v \in L^1(\R^3)\cap L^2(\R^3)$. For all $\psi \in L^1(\R^3)$, 
\begin{equation}\label{frakKL2}
 \|\mathfrak{K}_e\psi\|_2  \leqslant   \frac{1}{\pi}(2e)^{-1/4}\|\psi\|_1\ ,
 \end{equation}
 and for all $\phi,\psi\in L^1(\R^3)\cap L^2(\R^3)$
\begin{equation}\label{frakKprops2}
  \int_{\R^3} dx \phi(x) ( \mathfrak{K}_e \psi)(x) =   \int_{\R^3} dx (\mathfrak{K}_e  \phi )(x)\  \psi(x) \ .
 \end{equation}
and \begin{equation}\label{frakKprops4}
e\mapsto \int_{\R^3} \phi(x) ( \mathfrak{K}_e \psi)(x)\ dx
 \end{equation}
  is continuous. Finally, for all $x$, 
\begin{equation}\label{frakKprops3}
0 \leqslant   \mathfrak{K}_ev(x) \leqslant 1 \ .
 \end{equation}
 \end{lemma} 
 
 Then, as a direct consequence of \eqref{upform}, Lemma~\ref{ul2b} and the bound on $\rho'$ provided by Theorem~\ref{Mon} we have:
 \begin{lemma}\label{upbn}
 There is a constant independent of $C$ such that  for all $e$,
 $$
 \|u'\|_2 \leqslant C\rho^{-1}e^{-1/4}\ .
 $$
 \end{lemma}

 In the course of proving Theorem~\ref{theo:convexity} on convexity, we will need a bound on 
 \begin{equation}\label{toughnut}
\rho^{2}\int_{\R^3}  (\mathfrak{K}_e v) u'*u'\ dx
 \end{equation}
 which shows that for small $\rho$, this is negligible compared to $\rho^{-2}$. 
 By Young's inequality for convolutions, if we have  bounds on $\|\mathfrak{K}_e v\|_p$ and $\|u'\|_q$ with $1/p + 2/q = 2$, we can bound the integral in \eqref{toughnut}. As we shall see below, since $v\geqslant 0$, $\mathfrak{K}_e v$ can decay at infinity no faster than $|x|^{-2}$, and hence cannot belong to $L^p$ for $p \leqslant 3/2$.  Therefore, we will need to have a bound on $\|u'\|_q$ for fairly small $q$.  We shall see that $\|u'\|_q < \infty$ for all $q>1$ (see Theorem~\ref{Lpu'}), and we shall obtain a bound on $\|u'\|_{4/3}$ that can be combined with our bound on  $\|\mathfrak{K}_e v\|_2$ to obtain the necessary control on the integral in \eqref{toughnut}. 
 
 To do this, we need something more incisive than the bound \eqref{facA}.  We shall show (see section\-~\ref{sec:uprime}) that
$\mathfrak{Y}_e$, factors as the product of three commuting operators operators
 \begin{equation}\label{3fac}
 \mathfrak{Y}_e  =  (-\Delta)^{-1/2} (-\Delta+  8e)^{-1/2} [I +  \mathfrak{H}_e]
 \end{equation}
 where $\mathfrak{H}_e$ is convolution by an $L^1$ function with $L^1$ norm bounded by a constant multiple of $e^{1/2}$.  Hence $\mathfrak{H}_e$ is bounded on $L^p$ for all $p$ with a norm bounded by a multiple of $e^{1/2}$. Likewise,
$ (-\Delta+  4e)^{-1/2} $  is bounded on $L^p$ for all $p$ with a norm bounded by a multiple of $e^{-1/2}$.  Thus  $\mathfrak{Y}_e$ inherits  the $L^p$ to $L^q$ mapping properties of  $(-\Delta)^{-1/2}$, and  these are given by the Hardy-Littlewood-Sobolev (HLS) inequality.  In particular, this implies that there is a constant $C$ independent of $e\leqslant e_\star$ such that for all $1 < p < q < \infty$ with $1/p = 1/q - 1/3$,  
$$\|\mathfrak{K}_e \psi\|_q \leqslant C e^{-1/2} \|\psi\|_p\ .$$
Of course, since $\mathfrak{K}_e$ is not scale invariant, it satisfies further $L^p$ to $L^q$ bounds beyond those supplied by HLS, as we have already in Lemma~\ref{frakKprops}. However, this line of argument can only provide a bound on 
$\|\mathfrak{K}_e \psi\|_q$ for $q > 3/2$, and hence using this and $  u'=\mathfrak K_e(-4u+2\rho u\ast u+2\rho'u\ast u)$ can only provide bounds on $\|u'\|_q$ for $q > 3/2$.  To get down to   $\|u'\|_{4/3}$ and below, we need another self-referential formula for $u'$ with is 
\begin{equation}\label{upform2}
u' =  \mathfrak Y_e \psi \qquad{\rm where}\qquad \psi = 2\rho u*u -4u + 2e\rho' u*u - vu'\ .
\end{equation}
The merit of this formula is that, as we shall see, $\int_{\R^3} \psi\ dx =0$.    Recall that 
$$
(-\Delta)^{-1/2}\varphi =  \frac{1}{2\pi^2}\int |x-y|^{-2} \varphi(y) \ dy
$$
so that for $\varphi>0$, one can have  at best that $(-\Delta)^{-1/2}\varphi(x)$ decays at infinity like $|x|^{-2}$, and hence cannot belong to $L^q$ for $q \leqslant 3/2$. However, when $\varphi$ integrates to zero and decays sufficiently rapidly at infinity, 
there will be a cancellation so that  $(-\Delta)^{-1/2}\varphi(x)$ will decay more rapidly, up to as fast as $|x|^{-3}$.    We are therefore led to prove a variant of the HLS inequality for functions that integrate to zero, and of course using a norm on the input that is not scale invariant, but which measures the rate of decay at infinity.  This may be of wider utility, and  we carry this out in dimension $d$ for arbitrary $d$.

The norm we use on the input is built using the Lorentz norms $L_{p,q}$.  These are recalled in some more detail below, but recall that $L_{p,\infty}$ is weak $L^p$, $L_{p,p}$ is $L^p$, and $L_{p,1}$ is a strict subset of $L^p$.
For $0 < \beta < d$, let $\G_\beta$ denote the operator
\begin{equation}\label{Gdef0}
\G_\beta f(x) =\int_{\R^d}|x-y|^{-\beta}f(y) dy\ .
\end{equation}

\begin{definition}\label{newndef}  Let $f$  be a function such that $(1+|x|)^s f(x)\in L^1$, and such that $f\in L_{d/(d-\beta),1}$ ($L_{d/(d-\beta),1}$ is the Lorentz space with indices $d/(d-\beta),1$). 
Let $f_{\leqslant R}$ denote $f$ multiplied by the indicator function of the  closed ball of radius $R$, and let $f_{>R} := f - f_{\leqslant R}$.  
Given $s>d$, define
\begin{equation}\label{newnorm}
|\! |\! | f |\!|\!|_{\beta,s} =  \int_{\R^d} (1+|x|)^{s-d} |f(x)| dx +  \sup_{R>0} (1 +R)^{\beta+s - d}\|f_{>R}\|_{d/(d-\beta),1} 
\end{equation}
Define the space $\mathcal{L}_{\beta,s}$ to be the space of all measurable functions $f$ for which $|\! |\! | f |\!|\!|_{\beta,s} < \infty$. 
\end{definition} 

 We show below that if $f$ satisfies the bound 
\begin{equation}\label{isias1}
|f(x)| \leqslant M (1+ |x|)^{-r}\ ,
\end{equation}
then $f\in  \mathcal{L}_{\beta,s}$ for all $s < r$.   Then by Theorem~\ref{theo:pointwise}, we shall be able to apply the following theorem  with $s$ arbitrarily close to $4$, granted $v$ decays sufficiently rapidly so that $v u'\in \mathcal{L}_{\beta,s}$. Taking $R= 0$ in \eqref{newnorm}, we see that $\|f\|_{d/(d-\beta),1} \leqslant |\! |\! | f |\!|\!|_{\beta,s}$, and since $L^{d/(d-\beta)}  \subset L_{d/(d-\beta),1}$, $L^{d/(d-\beta)} \subset   \mathcal{L}_{\beta,s}$.
Evidently, $L^1 \subset \mathcal{L}_{\beta,s}$. Thus, $\mathcal{L}_{\beta,s} \subset L^1 \cap  L^{d/(d-\beta)}$, and hence 
for all $f\in \mathcal{L}_{\beta,s}$ with $\beta$, $s$ as specified, $f\in L^p$ for all $1\leqslant p \leqslant d/(d-\beta)$; i.e., the whole HLS interval  including the endpoints.

\begin{theorem}\label{HLSvar0}  Let $f\in  \mathcal{L}_{\beta,s}$ for some $d+1 \geq  s>d$  satisfying $\int_{\R^d} f(x)dx  =0$.  
 Then for all  $q \leqslant d/\beta$  such that 
 \begin{equation}\label{HLSV1}
q > \frac{d}{\beta+s -d}\ ,
\end{equation}
there is a constant $C$ depending only on $q$, $s$ and $\beta$ such that 
\begin{equation}\label{HLSV2}
\|\G_\beta f\|_{q} \leqslant C   |\! |\! | f |\!|\!|_{\beta,s} 
\end{equation}
\end{theorem}

With $s$ taken sufficiently close to $4$ and $\beta =2$ and $d=3$, we can get control on $\|\G_2f\|_q$ for $q$ arbitrarily close to $1$.   In this way we prove:

\begin{theorem}\label{Lpu'}  Let $e_\star$ be defined as in Theorem~\ref{Mon}. 
Assume that $(1+ |x|^4)v(x)\in L^1(\mathbb R^3)\cap L^2(\mathbb R^3)$ and that $(1+|x|^3)v(x) \in L^8(\R^3)$.
For all $p>1$,  there is a constant $C$ depending only on $p$  such that for all $e \leqslant e_\star$, 
$$\|u'\|_p  \leqslant  C e^{-3/2}\ .
$$
\end{theorem}
This provides  the control on $\|u'\|_q$ that we need to prove the theorem on convexity.

\begin{remark}
  As stated at the very beginning of the paper, we assume that $v$ is spherically symmetric.
  This is, however, used very little in the proofs.
  In fact, the only theorem that relies on the spherical symmetry is theorem\-~\ref{theo:pointwise}.
  We believe it should still hold (provided the decay constant in\-~(\ref{udecay}) is suitably adapted) without the spherical symmetry.
  In this case, the other theorems would not require the spherical symmetry.
\end{remark}

\section{Pointwise bounds on $u(x)$ -- Proof of Theorem \ref{theo:pointwise}}
  Let
  \begin{equation}
    \kappa:=\frac{|k|}{2\sqrt e}
    \label{kappa}
  \end{equation}
  in terms of which (\ref{uhat}) becomes
  \begin{equation}
    \rho\widehat u=(\kappa^2+1)\left(1-\sqrt{1-\frac{\frac\rho{2e}\widehat S}{(\kappa^2+1)^2}}\right)
    .
    \label{uhat_factorized}
  \end{equation}
  For small $\kappa$, since $x^4v$ is integrable, $\widehat S$ is $\mathcal C^4$
  \begin{equation}
    \frac\rho{2e}\widehat S=1-\beta \kappa^2+O(e^2\kappa^4)
    \label{expS}
  \end{equation}
  and $\beta$ is defined in\-~(\ref{betadef}):
  \begin{equation}\label{betabound}
    \beta  =   -\frac\rho{4e}\partial_\kappa^2 \widehat S \leqslant \rho\|x^2v\|_1
    .
  \end{equation}
  Therefore, defining
  \begin{equation}
    \widehat U_1:=(\kappa^2+1)^{-2}\left(1-\sqrt{1-\frac{(1-\beta \kappa^2)}{(\kappa^2+1)^2}}\right)
  \end{equation}
  $\widehat U_1$ coincides with $\widehat u$ asymptotically as $\kappa\to0$
   and we chose the prefactor $(\kappa^2+1)^{-2}$ in such a way that  $\widehat U_1$ is integrable.
  Define the remainder term
  \begin{equation}
    \widehat U_2:=\rho\widehat u-\widehat U_1
    =(\kappa^2+1)\left(1-\sqrt{1-2\zeta_1}\right)-(\kappa^2+1)^{-2}\left(1-\sqrt{1-2\zeta_2}\right)
  \end{equation}
  with
  \begin{equation}
    \zeta_1:=\frac{\frac\rho{4e}\widehat S}{(\kappa^2+1)^2}
    ,\quad
    \zeta_2:=\frac{1-\beta \kappa^2}{2(\kappa^2+1)^2}
    .
    \label{zetas}
  \end{equation}
  The rest of the proof proceeds as follows: we show that the Fourier transform of $\widehat U_1$ decays like $|x|^{-4}$ by direct analysis, then we show that $\Delta^2\widehat U_2$ is integrable and square integrable, which implies that it is subdominant as $|x|\to\infty$.

  \point We compute $U_1(x):=\int\frac{dk}{(2\pi)^3}e^{-ikx}\widehat U_1(k)$.
  We write
  \begin{equation}
  \sqrt{1 - \frac{1-\beta \kappa^2}{(1+\kappa^2)^2}}  = \frac{\kappa}{1+\kappa^2} \sqrt{2 + \beta + \kappa^2}  = \frac{1}{\pi}  \frac{|\kappa|(2 + \beta + \kappa^2)}{1+\kappa^2} \int_0^\infty \frac{1}{2+\beta + t + \kappa^2} t^{-1/2} dt \ \ .
  \end{equation}
  Therefore,
  \begin{equation}
  \widehat{U}_1 := (\kappa^2+1)^{-2}  -  \frac{\kappa}{\pi} (\kappa^2+1)^{-2} \left(1 + (\beta+1)\frac{1}{1+\kappa^2}\right)  \int_0^\infty \frac{1}{2+\beta + t + \kappa^2} t^{-1/2} dt 
  .
  \end{equation}
  We take the inverse Fourier transform of $\widehat U_1$, recalling the definition of $\kappa$\-~(\ref{kappa})
  \begin{equation}
    U_1(x)=\frac{e^{\frac32}}\pi e^{-2\sqrt e|x|}
    -\frac1\pi \left(\delta(x)+\frac{(\beta+1)e}\pi\frac{e^{-2\sqrt e|x|}}{|x|}\right)
    \ast f_1\ast f_2
    \label{U1}
  \end{equation}
  where
  \begin{equation}
    f_1(x):=\frac{e^{\frac32}}{\pi^3}\int dk\ e^{-ik(2\sqrt ex)}\frac{|k|}{(k^2+1)^2}
  \end{equation}
  and
  \begin{equation}
    f_2(x):=\frac{e^{\frac32}}{\pi^3}\int dk\ e^{-ik(2\sqrt ex)}\int_0^\infty \frac{dt}{\sqrt t}\frac1{2+\beta+t+k^2}
    =\frac{e}{\pi|x|}  \int_0^\infty e^{-\sqrt{2+\beta + t}(2\sqrt{e}|x|)}  t^{-1/2} dt\ ,
  \end{equation}
  now, for all $T>0$,
  \begin{eqnarray}
  \int_0^\infty e^{-\sqrt{2+\beta + t}(2\sqrt{e}|x|)}  t^{-1/2}\ dt &=& \int_0^{T} e^{-\sqrt{2+\beta + t}(2\sqrt{e}|x|)}  t^{-1/2}+\int_{T}^\infty e^{-\sqrt{2+\beta + t}(2\sqrt{e}|x|)}  t^{-1/2}\ dt
  \nonumber\\
  &\leqslant&
  2 T^{1/2} e^{-\sqrt{2+\beta}(2\sqrt{e}|x|)}  +  \frac{1}{\sqrt{e}|x|}e^{-\sqrt{T}(2\sqrt{e}|x|)}\ .
  \label{boundf2}
  \end{eqnarray}
  Choosing $T = 2+\beta$, we see that for large $(2\sqrt{e}|x|)$, $0 \leqslant f_2(x) \leqslant C e^{-\sqrt{2+\beta}(2\sqrt{e}|x|)}$.
  Furthermore,
  \begin{equation}
    f_1(x)=\frac{e^{\frac32}}{\pi^3}\int dk\ e^{-ik(2\sqrt ex)}\frac1{|k|}\frac{k^2}{(k^2+1)^2}
    =
    \frac{e^{\frac32}}{\pi^3}\frac1{|x|^2}\ast g
    ,\quad
    g(x)=\frac{(1-\sqrt e|x|)e^{-(2\sqrt e)|x|}}{|x|}
    \label{g}
  \end{equation}
  Using 
  \begin{equation}
  \frac{1}{|x-y|^2} = \frac{1}{|x|^2} +  \frac{-|y|^2 + 2x\cdot y}{|x|^2|x-y|^2}
  \end{equation}
  twice and the fact that $g(y)$ is even, integrates to zero, and $\int y g(y)\ dy=0$,
  \begin{equation}
  f_1(x) = \frac{1}{|x|^4}  \frac{e^{\frac32}}{\pi^3}\left( - \int_{\R^3} |y|^2 g(y){\rm d} y  +   \int_{\R^3}  \frac{(-|y|^2 + 2x\cdot y)^2}{|x-y|^2}g(y) {\rm d} y\right)
  \label{f1}
  \end{equation}
  We compute
  $
   \int_{\R^3} |y|^2 g(y){\rm d} y = -\frac{3\pi}{2e^2}
  $, and then using the symmetry of $g$ once more, 
  \begin{equation}
  \lim_{|x|\to \infty}  \int_{\R^3}  \frac{(x\cdot y)^2}{ |x-y|^2}g(y) {\rm d} y = \frac{1}{3}\int_{\R^3} |y|^2 g(y){\rm d} y  = -\frac\pi{2e^2}\ ,
  \end{equation}
  Therefore,
  \begin{equation}
  \lim_{|x|\to \infty} |x|^4 f_1(x) =  -\frac{1}{2\pi^2\sqrt e} \qquad{\rm and}\qquad   \lim_{|x|\to \infty} |x|^4 U_1(x)  =  \frac{1}{2\pi^2\sqrt e}\sqrt{2+\beta}\ .
  \label{U1decay}
  \end{equation}
  \bigskip

  \indent
  We now turn to an upper bound of $U_1$.
  First of all, if $|x|\leqslant\frac1{\sqrt e}$, then by\-~(\ref{g}) and\-~(\ref{f1}),
  \begin{equation}
    f_1(x)\geqslant 0
  \end{equation}
  and if $|x|>\frac1{\sqrt e}$, then
  \begin{equation}
     f_1(x)\geqslant
     -\frac{1}{|x|^4}  \frac{e^{2}}{\pi^3}\int_{\mathbb R^3}  \frac{(-|y|^2 + 2x\cdot y)^2}{ |x-y|^2}e^{-(2\sqrt e)|y|} {\rm d} y
     .
  \end{equation}
  We split the integral into two parts: $|y-x|>|x|$ and $|y-x|<|x|$.
  We have, (recalling $|x|>\frac1{\sqrt e}$),
  \begin{equation}
     \int_{|y-x|>|x|}  \frac{(-|y|^2 + 2x\cdot y)^2}{ |x-y|^2}e^{-(2\sqrt e)|y|} {\rm d} y
     \leqslant
     e^{-\frac52}C
  \end{equation}
  for some constant $C$ (we use a notation where the constant $C$ may change from one line to the next).
  Now,
  \begin{equation}
     \int_{|y-x|<|x|}  \frac{(-|y|^2 + 2x\cdot y)^2}{ |x-y|^2}e^{-(2\sqrt e)|y|} {\rm d} y
     \leqslant
     e^{-\sqrt e|x|}\int_{|y-x|<|x|}  \frac{(|y|^2 + 2|x||y|)^2}{ |x-y|^2} {\rm d} y
     \leqslant
     |x|^5e^{-\sqrt e|x|}C
     .
  \end{equation}
  Therefore, for all $x$,
  \begin{equation}
    f_1(x)\geqslant
    -\frac{1}{|x|^4}C(e^{-\frac12}+e^2|x|^4e^{-\sqrt e|x|})
    .
  \end{equation}
  Finally, by use\-~(\ref{boundf2}),
  \begin{equation}
    |x|^4\left(\delta(x)+\frac{(\beta+1)e}{\pi}\frac{e^{-2\sqrt e|x|}}{|x|}\right)\ast f_1\ast f_2(x)\geqslant
    -Ce^{-\frac12}
    .
  \end{equation}
  All in all, by\-~(\ref{U1}), (since $|x|^4e^{\frac32}e^{-2\sqrt e|x|}<Ce^{-\frac12}$)
  \begin{equation}
    |x|^4U_1(x)\leqslant Ce^{-\frac12}
    .
    \label{boundU1}
  \end{equation}
  \bigskip

  \point
  We now show that $\Delta^2\widehat U_2$ is integrable and square-integrable.
  We use the fact that
  \begin{equation}
    16e^2\Delta^2\equiv\partial_\kappa^4+\frac4\kappa\partial_\kappa^3
    .
    \label{D2}
  \end{equation}
  We have, by the Leibniz rule,
  \begin{equation}
    \partial_\kappa^n\widehat U_2=
    \sum_{i=0}^n{n\choose i}\left(
      \partial_\kappa^{n-i}(\kappa^2+1)\partial_\kappa^i(1-\sqrt{1-2\zeta_1})
      -
      \partial_\kappa^{n-i}(\kappa^2+1)^{-2}\partial_\kappa^i(1-\sqrt{1-2\zeta_2})
    \right)
    .
    \label{leibnitz}
  \end{equation}
  Furthermore,
  \begin{equation}
    \partial_\kappa^n(1-\sqrt{1-2\zeta_j})
    =
    \sum_{p=1}^n\partial_{\zeta_j}^p(1-\sqrt{1-2\zeta_j})\sum_{\displaystyle\mathop{\scriptstyle l_1,\cdots,l_p\in\{1,\cdots,n\}}_{l_1+\cdots+l_p=n}}
    c^{(p,n)}_{l_1,\cdots,l_p}\prod_{i=1}^n\partial_\kappa^{l_i}\zeta_j
    \label{dsqrt}
  \end{equation}
  for some family of constants $c_{l_1,\cdots,l_p}^{(p,n)}$ which can easily be computed explicitly, but this is not needed.
  Now, since $S\geqslant 0$, $\frac\rho{1e}|\widehat S|\leqslant 1$, so $|\zeta_1|\leqslant\frac12$ and $\zeta_1=\frac12$ if and only if $\kappa=0$.
  Therefore, $\widehat U_2$ is bounded when $\kappa$ is away from 0, so it suffices to show that $\Delta^2\widehat U_2$ is integrable and square integrable at infinity and at 0.
  \bigskip

  \subpoint We first consider the behavior at infinity, and assume that $\kappa$ is sufficiently large.
  The fact that $\partial_\kappa^{n-i}(\kappa^2+1)^{-2}\partial_\kappa^i(1-\sqrt{1-2\zeta_2})$ is integrable and square integrable at infinity follows immediately from\-~(\ref{zetas}).
  To prove the corresponding claim for $\zeta_1$, we use the fact that $|x|^4 v$ square integrable, which implies that $\widehat S$ is as well.
  Therefore, by\-~(\ref{zetas}) for $0\leqslant n\leqslant 4$, $\kappa^2\partial_\kappa^n\zeta_1$ is integrable at infinity, and, therefore, square-integrable at infinity.
  Furthermore, by\-~(\ref{zetas}), $\zeta_1<\frac12-\epsilon$ for large $\kappa$, and $\partial^n\zeta_1$ is bounded, so $\partial_\kappa^{n-i}(\kappa^2+1)\partial_\kappa^i(1-\sqrt{1-2\zeta_1})$ is integrable and square integrable.
  \bigskip

  \subpoint As $\kappa\to0$
  \begin{equation}
    \zeta_i=\frac12(1-(\beta+2)\kappa^2)+O(\kappa^4)
  \end{equation}
  and since $\beta\geqslant 0$,
  \begin{equation}
    1-2\zeta_i\geqslant \kappa^2+O(\kappa^4)
    .
  \end{equation}
  therefore, for $p\geqslant 1$
  \begin{equation}
    \partial_{\zeta_j}^p(1-\sqrt{1-2\zeta_j})=O(\kappa^{1-2p})
  \end{equation}
  and, since $\zeta_i$ is $\mathcal C^4$, for $3\leqslant n\leqslant 4$,
  \begin{equation}
    \partial\zeta_i=-(\beta+2)\kappa+O(\kappa^3)
    ,\quad
    \partial^2\zeta_i=-(\beta+2)+O(\kappa^2)
    ,\quad
    \partial^n\zeta_i=O(\kappa^{4-n})
    .
  \end{equation}
  Therefore, for $1\leqslant i\leqslant 4$, by\-~(\ref{dsqrt})
  \begin{equation}
    \partial_\kappa^i(1-\sqrt{1-2\zeta_1})
    -
    \partial_\kappa^i(1-\sqrt{1-2\zeta_2})
    =
    O(\kappa^{3-i})
  \end{equation}
  and
  \begin{equation}
    \partial_\kappa^i(1-\sqrt{1-2\zeta_1})
    =O(\kappa^{1-i})
    ,\quad
    \partial_\kappa^i(1-\sqrt{1-2\zeta_2})
    =O(\kappa^{1-i})
    .
  \end{equation}
  Thus, by\-~(\ref{leibnitz}), as $\kappa\to0$,
  \begin{equation}
    |\partial_\kappa^4\widehat U_2|=O(\kappa^{-1})
    ,\quad
    \frac4\kappa|\partial_\kappa^3\widehat U_2|=O(\kappa^{-1})
    .
  \end{equation}
  Thus, $\Delta^2\widehat U_2$ is integrable and square integrable.
  And since the $O(\cdot)$ hold uniformly in $e$ on all compact sets, by\-~(\ref{D2}),
  \begin{equation}
    |x|^4U_2(x)
    \leqslant
    \frac{8e^{\frac32}}{16e^2}\int \left(\partial_{|k|}^4+\frac4{|k|}\partial_{|k|}^3\right)\hat U_2(|k|)\ dk
    \leqslant\frac C{\sqrt e}
    .
  \end{equation}
  This along with\-~(\ref{U1decay}) and\-~(\ref{boundU1}) implies\-~(\ref{udecay}) and\-~(\ref{fourdecay}).
\qed

\section{Monotonicity  of $\rho(e)$ -- Proof of Theorem ~\ref{Mon}}\label{sec:monotonicity}

The proof  uses certain estimates whose simple proofs are provided in Section~\ref{UB}, and it also relies on the following preliminary result:

\begin{lemma}\label{parmon}  The function $e \mapsto e\rho(e)$ is strictly  monotone increasing.
\end{lemma}

\begin{proof} Suppose that for some $\tilde e > e$, $\tilde s \rho(\tilde e) \leqslant e \rho(e)$.  Define the operator $K_e = (-\Delta + v + 4e)^{-1}$ and likewise $K_{\tilde e}$.  We use a variant of the iterative scheme used in \cite{CJL20}. First, write the simple equation~(\ref{simp}) as a fixed point equation:

\begin{equation}\label{premon}
u(e,x) = K_ev(x) + 2 e \rho(e) K_e(u*u)(x)\ 
\end{equation}
as in \cite{CJL20}.  
Next, inductively define a sequence of functions as follows:
\begin{equation}\label{premon2}
u_n(e,x) = K_ev(x) + 2 e \rho(e) K_e(u_{n-1}*u_{n-1})(x)\quad{\rm for}\ n\geq 1 \quad{\rm with}\quad u_0(e,x) =0\ .
\end{equation}
A simple induction shows that for all $n \geq 1$
$$
0 \leqslant u_{n-1}(e,x)  < u_n(e,x) < u(e,x)
$$
Then by dominated convergence,  $\lim_{n\to\infty}u_n(e,x)$ exists and is integrable and satisfies \eqref{premon}.  The iteration in \eqref{premon2} differs from the one used to construct $u(e,x)$ and $\rho(e)$ in \cite{CJL20} in that we are now using the function $\rho(e)$ constructed there, while in \cite{CJL20}, we had to use an increasing sequence $\rho_n(e)$ of minorants to it.  Since for each $n$, $\rho_n(e) \leqslant \rho(e)$, it follows that for each $n$, the function $u_n(e,x)$ is pointwise larger than the corresponding $n$th term in the approximating sequence constructed in \cite{CJL20}, Since that sequence was shown to converge to $u(e,x)$, it follows that so does the sequence constructed here. That is:
$$
u(e,x) = \lim_{n\to\infty}u_n(e,x)
$$
Now the integral kernel for $K_e$ is monotone decreasing in $e$. Therefore, a simple induction shows that if  $\tilde e > e$, $\tilde e \rho(\tilde e) < e \rho(e)$,  then 
$u_n(\tilde e,x)  \leqslant u_n(e,x)$
for all $n$, and hence
$$u(\tilde e,x) \leqslant u(e,x)\ .$$
Integrating we find that 
${\displaystyle \frac{1}{\rho(\tilde e)} \leqslant \frac{1}{\rho(e) }}$,
and this leads to $\tilde e \rho(\tilde e) > e \rho(e)$ which contradicts our hypothesis. 
\end{proof}

\begin{proof}[Proof of Theorem ~\ref{Mon}] Suppose for the moment that both $\rho(e)$ and $u(e,x)$ are differentiable in $e$ and define
\begin{equation}\label{simp2}
\rho'(e) = \frac{{\rm d}}{{\rm d}e}\rho(e) \qquad{\rm and}\qquad u'(e,x) = \frac{\partial}{\partial e}u(e,x)\ ;
\end{equation}
we shall come back and justify this later.
Differentiating \eqref{simp} in $e$, we find
\begin{equation}\label{simpd}
  u'=\mathfrak K_e(-4u+2\rho u\ast u+2\rho'u\ast u)
\end{equation}
where  $\mathfrak K_e$  is given by  \eqref{bg6}. Combining this with \eqref{bg6X} and \eqref{frakKprops2} yields, as explained in the introduction,
\begin{equation}\label{rpratio}
\frac{e}{\rho} \rho' = \frac { 1 + \rho\int  (   \mathfrak{K}_e v)(\rho u*u - 2u)dx }{1 -\rho^2\int (   \mathfrak{K}_e v) u*u dx   }\ .
\end{equation}

By Lemma~\ref{frakKprops}, $0 \leqslant    \mathfrak{K}_ev(x) \leqslant 1$, with strict inequality  for $x$ sufficiently large, and hence, by~(\ref{intu})
$$
\int (   \mathfrak{K}_e v) u*u dx < \int  u*u dx = \frac{1}{\rho^2}\ .
$$
Consequently,
$$
1 - {\rho^2}\int (   \mathfrak{K}_e v) u*u dx > 0\ ,
$$
so that the denominator in \eqref{rpratio} is positive, and  hence, granted our differentiability assumptions,  we shall have proved that $\rho' > 0$  once we have proved that the numerator is positive; i.e., that
$$
1 + \rho \int  (   \mathfrak{K}_e v)(\rho u*u - 2u)dx >0\ .
$$
One might hope to use Lemma~\ref{kl1} once more, as we did for the denominator in \eqref{rpratio}. If it were true that 
\begin{equation}\label{gb}
\rho u*u - 2u \leqslant 0
\end{equation}
everywhere,  this would be immediate. The explicit solution provided in Section \ref{sec:explicit} has this property, and it appears to be true in cases that we have examined numerically. However, we lack an analytic proof, and  must resort to estimates that yield the desired conclusion but only when $e$ is sufficiently small or large.

We have 
\begin{equation}\label{gb2}
1 + \rho\int  (   \mathfrak{K}_e v)(\rho u*u - 2u)dx  \geq  1 - 2\rho \|   \mathfrak{K}_e v\|_2\|u\|_2,
\end{equation}
and we need only show that for all $e$ sufficiently small $2\rho \|   \mathfrak{K}_e v\|_2\|u\|_2 \leqslant 1$, and, under that additional assumption that $\|v\|_\infty < \infty$, the same is true for all $e$ sufficiently large.  We first consider small $e$.

By Lemmas~\ref{ul2b} and \ref{frakKprops},
$$
2\|   \mathfrak{K}_e v\|_2\|u\|_2   \leqslant \frac1{2^{\frac74}\pi^{\frac32}}\|v\|_1^2  (2e)^{-\frac12}  
$$
By \eqref{con4B},   $\rho \leqslant \frac{4e}{\|v\|_1}$, and hence 
$$
2\rho \|   \mathfrak{K}_e v\|_2\|u\|_2 \leqslant \frac1{2^{\frac14} \pi^{\frac32}}\|v\|_1 e^{\frac12}\ .
$$
Hence $\rho'$ is positive when
${\displaystyle
e < \frac{\sqrt 2 \pi^{3}}{\|v\|_1^2}}$. Moreover,  since for all such $e$,   $2\rho \|   \mathfrak{K}_e v\|_2\|u\|_2 \leqslant 1$,
$$
{\rho^2}\int (   \mathfrak{K}_e v) u*u   = \int[ (   \rho \mathfrak{K}_e v)*u  ]\rho u  \leqslant \rho  \|   \mathfrak{K}_e v\|_2\|u\|_2  \leqslant \tfrac12\ ,
$$
we have from \eqref{rpratio} and \eqref{con4B} that
\begin{equation}\label{rhopb}
\rho'  \leqslant 4 \frac{\rho}{e} \leqslant   \frac{16}{\|v\|_1}\ ,
\end{equation}
and this proves \eqref{rhopb2}.

We next consider large $e$. By \eqref{con4B} and Lemma~\ref{highdens1},
$$
\rho \|u\|_2  \leqslant  \frac{4e}{\|v\|_1} \frac{\|v\|_2}{2e} = 2\frac{\|v\|_2}{\|v\|_1}\ .
$$
By Lemma~\ref{frakKprops} once more,    
$$ 2\rho \|   \mathfrak{K}_e v\|_2\|u\|_2  \leqslant  \frac{1}{\pi}(2e)^{-\frac14}\|v\|_1 \left(\frac{2\|v\|_2}{\|v\|_1}\right)  \leqslant  \frac{2^{\frac34}}{\pi} \|v\|_2e^{-\frac14}\ .$$
Thus, $\rho(e)$ is strictly monotone also for ${\displaystyle e > \frac{2^3\|v\|_2^4}{\pi^4}}$.

We now deal with the differentiability assumptions by first considering finite differences. Fix $\tilde e > e > 0$ and let $\tilde u(x) \equiv u(\tilde e,x)$ and let $u(x) \equiv u(e,x)$.  Likewise let $\tilde \rho \equiv \rho(\tilde e)$ and $\rho \equiv \rho(e)$.  Finally, define $\delta e:= \tilde e -e$, $\delta u:= \tilde u -u$ and $\delta \rho:= \tilde \rho - \rho$.
Using the identity $\tilde a \tilde b -ab = \tilde a (\tilde b -b) + b(\tilde a -a)$ repeatedly, we find
$$(-\Delta + v + 4\tilde e) \delta u + 4 u \delta e  =  2(\tilde e \delta \rho + \rho \delta \tilde e)\tilde u * \tilde u +4e\rho \left(\frac{u+\tilde u}{2}\right)*\delta u\ .$$
Define the operator
\begin{equation}\label{sim81}
\widetilde{   \mathfrak{K}_e }  = \left(-\Delta + v + 4\tilde e - 2e\rho C_u - 2e\rho C_{\tilde u}\right)^{-1}\ .
\end{equation}
Now,
$$
2e\rho C_{\tilde u}  = 2\frac{e\rho}{\tilde{e}\tilde{\rho}} \tilde {e}\tilde{\rho}C_{\tilde u},
$$
and since $e\mapsto e\rho(e)$ is monotone by Lemma~\ref{parmon},
$$
\|2e\rho C_{\tilde u}\| \leqslant  \|2\tilde {e}\tilde{\rho}  C_{\tilde u}\| \leqslant 2\tilde{e}
$$
where the norm is the operator norm on  $L^p(\R^3)$  for any $p$.  Thus, the operator $\widetilde{   \mathfrak{K}_e }$ is even somewhat better behaved than  $\mathfrak{K}_e$; it is bounded on $L^2(\R^3)$ as well as being bounded from $L^1(\R^3)$ to $L^2(\R^3)$, although the former bound deteriorates as $\tilde{e} \downarrow e$. However, the latter bound persists:
As noted in Remark~\ref{frakYbndR},   it is easy to see that  the bound of Lemma~\ref{frakKprops} holds also for $\widetilde{   \mathfrak{K}_e }$, by the same proof.  We now have 

\begin{equation}\label{sim82}
\delta u =  \widetilde{   \mathfrak{K}_e } (2(\tilde e \delta \rho + \rho \delta  e)\tilde u * \tilde u  - 4 u \delta e)\ .
\end{equation}
Multiplying by $v$ and integrating we see
$$
\int v \delta u dx  = \int (\widetilde{   \mathfrak{K}_e }v)(2(\tilde e \delta \rho + \rho \delta  e)\tilde u * \tilde u  - 4 u \delta e) dx
$$
Next, 
$$
\int v \delta u dx  =    \int v(1-u) dx  - \int v(1-\tilde u)dx = \frac{2e}{\rho} - \frac{2\tilde e}{\tilde \rho} = 2\tilde e\frac{\delta \rho}{\rho \tilde \rho} - 2\frac{\delta e}{\rho}\ .
$$
We then conclude
$$
\left(\frac{1}{\rho \tilde \rho}  -    \int (\widetilde{   \mathfrak{K}_e }v) \tilde u * \tilde u dx\right) \tilde e \delta \rho
= \delta e \left( \frac{1}{\rho}  +   \int  \widetilde{   \mathfrak{K}_e }v (\rho \tilde u * \tilde u - 2u) dx \right)\ ,
$$ 
and hence 
\begin{equation}\label{sim75}
\frac{\tilde e}{\tilde \rho}\frac{\delta \rho}{\delta e} =  \frac{1 + \rho \int (\widetilde{   \mathfrak{K}_e }v)(\rho \tilde u * \tilde u - 2u)dx}{1 -  \rho \tilde \rho\int (\widetilde{   \mathfrak{K}_e }v) \tilde u * \tilde u dx}
\end{equation}
which may be compared to \eqref{rpratio}.  

Theorem~\ref{theo:pointwise} says that, under the additional assumption that $(1+|x|^4)v(x)\in L^1(\R^3)\cap L^2(\R^3)$,  for any compact interval $[a,b]$, $\sup_{e\in[a,b]}u(x,e)$ is integrable. Then by Lebesgue's Dominated Convergence Theorem,
$e\mapsto u(x,e)$ is continuous into $L^1(\R^d)$.  We already know that   $e\mapsto \rho(e)$ is continuous by \cite[Theorem~1.3]{CJL20}, even without the additional assumption, though with this assumption,  it also follows from the pointwise continuity of $u(x,e)$ since $\rho(e) = (\int u(x,e){\rm d}x)^{-1}$. 
Now by Lemma~\ref{mfKcont} and Remark~\ref{mfKcontR}, it follows that 
$$\lim_{\tilde{e} \to e} \frac{1 - \rho \int (\widetilde{   \mathfrak{K}_e }v)(\rho \tilde u * \tilde u - 2u)dx}{1 -  \rho \tilde \rho\int (\widetilde{   \mathfrak{K}_e }v) \tilde u * \tilde u dx} =
 \frac{1 - \rho \int (  \mathfrak{K}_e v)(\rho  u * u - 2u)dx}{1 -  \rho \tilde \rho\int (  \mathfrak{K}_e v)  u *  u dx}\ ,$$
 and the right side is a continuous function of $e$.   It follows from \eqref{sim75} that $\rho(e)$ is continuously differentiable.

Finally, by \eqref{sim82},  
\begin{equation}
\frac{\delta u}{\delta e}  =  \widetilde{   \mathfrak{K}_e } \left(2(\tilde e \frac{\delta\rho}{\delta e} + \rho )\tilde u * \tilde u  - 4 u \right)\ .
\end{equation}
and now the limit $\tilde{e} \to e$ is controlled by Lemma~\ref{frakYbnd}  and Remark~\ref{frakYbndR},
yielding the proof of \eqref{simpd}. 
 \end{proof}

 \section{Convexity of $\rho e(\rho)$  -- Proof of theorem \ref{theo:convexity}}

In this section we prove the convexity of $\rho e(\rho)$ for small $\rho$.

First of all, we show that the convexity of $\rho\mapsto\rho e(\rho)$ is equivalent to the convexity of $e\mapsto\frac1{\rho(e)}$.
Let a prime  denote differentiation with respect to $e$ and a dot denote differentiation with respect to $\rho$. 
Assuming differentiability for now, we have $\dot e = (\rho')^{-1}$,  and $\ddot e = - (\rho')^{-3}\rho''$. Therefore,
\begin{equation}\label{convex1}
\frac{{\rm d}^2}{{\rm d}\rho^2}(\rho e(\rho)) = 2\dot e  + \rho \ddot e =  (\rho')^{-3} [2 (\rho')^{2} - \rho\rho'']
.
\end{equation}
Now one computes
\begin{equation}\label{convex2}
\frac{d^2}{de^2}\left(\frac{1}{\rho(e)}\right) =  \frac{2}{\rho^3}(\rho')^2 -\frac{1}{\rho^2}\rho''  =  \rho^{-3} [2 (\rho')^{2} - \rho\rho'']\ .
\end{equation}
Finally, by Theorem~\ref{Mon}, $\rho'>0$, so $\rho\mapsto \rho e(\rho)$ is convex if and only if $e\mapsto\frac1{\rho(e)}$ is.
We will now show that $\frac1\rho$ is a convex function of $e$.

\begin{proof}[Proof of Theorem~\ref{theo:convexity}]
As in the proof of the monotonicity, we begin by assuming the $u'$ and $\rho'$  are differentiable, and formally compute   $u''$ and $\rho''$.
Start from \eqref{simpd} and differentiate again to find
\begin{equation}\label{simpd2e}
(-\Delta + 4e + v)u''  + 4u' =  4\rho' u*u + 8\left(1 + \frac{e}{\rho}\rho'\right) \rho u*u'  + 4e\rho u'*u'  + 2e\rho'' u*u + 4e\rho u*u''\ .
\end{equation}
Therefore
\begin{equation}\label{simpd2b}
u''  = \mathfrak{K}_e\left(-4u' +  4\rho' u*u + 8\left(1 + \frac{e}{\rho}\rho'\right) \rho u*u'  + 4e\rho u'*u'  + 2e\rho'' u*u\right) \ .
\end{equation}
Multiplying by $v$ and integrating,
\begin{equation}\label{simpd2c}
\int  vu''dx =  \int  \mathfrak{K}_ev \left [  -4u' +  4\rho' u*u + 8\left(1 + \frac{e}{\rho}\rho'\right) \rho u*u'  + 4e\rho u'*u'  \right]  +  2e\rho'' \int (\mathfrak{K}_e v) u*u dx\ .
\end{equation}
By \eqref{convex2}, 
$$e\rho'' =  - e\rho^2\left(\frac{1}{\rho}\right)'' + \frac{2\rho}{e}\left(\frac{e}{\rho}\rho'\right)^2 \ ,
$$
and hence 
\begin{equation}\label{simpd2d}
2e\rho'' \int (\mathfrak{K}_e v) u*u dx  =  -2e\rho^2\int(\mathfrak{K}_e v) u*u dx \left(\frac{1}{\rho}\right)''  +  \frac{4\rho}{e}\left(\frac{e}{\rho}\rho'\right)^2   \int(\mathfrak{K}_e v) u*u dx
\end{equation}
Twice differentiating the second equation in \eqref{simp},
\begin{equation}\label{sdff}
2e\left(\frac{1}{\rho}\right)''  =  4\frac{\rho'}{\rho^2}+2\left(\frac e\rho\right)''
=\frac{4}{e\rho} \frac{e}{\rho}\rho' - \int v u'' dx \ .
\end{equation}
Then by the calculations above,
\begin{eqnarray}\label{fin}
\left[  2e - 2e\rho^2 \int (\mathfrak{K}_e v) u*u dx
 \right] \left(\frac{1}{\rho}\right)''  &=&   \frac{4}{e\rho} \frac{e}{\rho}\rho'  \nonumber \\
 &-& \int  \mathfrak{K}_e v \left [  -4u' +  \frac{4}{e}\left(\frac{e}{\rho}\rho'\right)\rho u*u + 8\left(1 + \frac{e}{\rho}\rho'\right) \rho u*u'  + 4e\rho u'*u'  \right]\nonumber\\
 &-&  \frac{4\rho}{e}\left(\frac{e}{\rho}\rho'\right)^2   \int(\mathfrak{K}_e v) u*u dx\ .
\end{eqnarray}
Note that, by Lemma~\ref{frakKprops} and\-~(\ref{intu}),
$$\left[  2e - 2e\rho^2 \int (\mathfrak{K}_e v) u*u dx
 \right] > 2e - 2e\rho^2 \int  u*u dx =0\ ,
 $$
 and hence if the right side of \eqref{fin} is non-negative, the convexity is proved.   This will be proved by showing that the largest term on the right is  $ \frac{4}{e\rho} \frac{e}{\rho}\rho' $ which is of order $\rho^{-2}$ for small  $\rho$, while all the others much smaller for small $\rho$.

We  require some estimates on $\|u'\|_p$ and in most instances the estimate on $\|u'\|_2$ provided by Lemma~\ref{upbn} suffices. 
For example,  by Lemma~\ref{upbn} and Lemma~\ref{frakYbnd}, 
$$
\left|\int (\mathfrak{K}_e v) u' dx \right| \leqslant \|\mathfrak{K}_e v\|_2\|u'\|_2  \leqslant C\rho^{-1}e^{-\frac12}\ 
$$
The others, except one, are similar  and the required  estimate on $\rho'$ is provided by Theorem ~\ref{Mon}.  The exceptional term is 
$$4e\rho \int (\mathfrak{K}_e v)u' *u' dx\ .$$
To handle this term we need a good estimate on $\|u'\|_p$ for $p< \frac43$ in order to use Young's inequality and our bound on $\|\mathfrak{K}_e v\|_2$. It is much harder to control $\|u'\|_p$ for $p$ smaller than $2$ than for $p$ greater than $2$.
Lemma~\ref{Lpu'} proved below says that 
for all $p>1$ and $e_0>0 $,  there is a constant $C$ depending only on $p$ and  $e_0$  such that for all $e \leqslant e_0$, 
$\|u'\|_p  \leqslant  C \rho^{-1}e^{-1/2}$. Then
by Young's inequality, we have
$$
\left|  4e\rho \int (\mathfrak{K}_e v)u' *u' dx \right|   \leqslant 4  e\rho \|\mathfrak{K}_e v\|_2 \|u'\|_{4/3}^2  \leqslant C e\rho e^{-1/4}(e^{-3/2})^2 = C \rho^{-1}e^{-1/4}\ .
$$
For small $e$, this is negligible compared to the main term, ${\displaystyle   \frac{4}{e\rho} \frac{e}{\rho}\rho' }$. 

To make this rigorous, we write out the same computation in finite differences as in the proof of monotonicity. This is straightforward, and left to the reader. 
\end{proof}

\section{The condensate fraction  -- Proof of Theorem~\ref{theo:condensate_eta}}\label{sec:condensate_fraction}

\point
Let us start by proving\-~(\ref{etaf}). 
Recall\-~(\ref{simpleq_eta})-(\ref{eta}):
\begin{equation}
  (-\Delta+2\mu+4e_\mu) u_\mu=(1-u_\mu)v+2\rho e_\mu u_\mu\ast u_\mu
  ,\quad
  e_\mu=\frac\rho2\int (1-u_\mu(x))v(x)\ dx
  \label{simpleq_eta2}
\end{equation}
and
\begin{equation}
  \eta=\partial_\mu e_\mu|_{\mu=0}
  .
  \label{eta2}
\end{equation}
Note that $e_0 = e$, and we write $u = u_0$ to denote the solution of the simple equation.
One can show the existence of a solution to this equation in a very similar way to the proof in\-~\cite{CJL20} that\-~(\ref{simp}) has a solution.
Furthermore, one can prove that $u_\mu$ is differentiable with respect to $\mu$ in the same way as in the proof of the differentiability of $\rho$ with respect to $e$ in Section\-~\ref{sec:monotonicity}.
The details of these two proofs are left to the reader.

\indent
Define 
\begin{equation}\label{sdef}
s := \partial _\mu u_\mu|_{\mu =0}\ .
\end{equation}  
Differentiating \eqref{simpleq_eta} in $\mu$ and setting $\mu =0$, one has
\begin{equation}\label{der}
(2 + 4\eta)u + (-\Delta +4e)s  = -sv + 4\rho e s*u + 2\rho \eta u*u\ .
\end{equation} 
Recalling the definition\-~(\ref{bg6}) of $\mathfrak K_e:=  (-\Delta + v + 4e( 1 - C_{\rho u}))^{-1}$, we have
\begin{equation}
  s = \mathfrak{K}_e (2\eta \rho u*u -2u - 4\eta u)
  .
\end{equation}
Furthermore, by\-~(\ref{eta2}) $\eta = - \frac{\rho}{2} \int  sv\ dx$, so
\begin{equation}
 \eta = - \frac{\rho}{2} \int v \mathfrak{K} (2\eta \rho u*u -2u - 4\eta u)\ dx
 .
\end{equation}
Solving for $\eta$ yields \eqref{etaf}, which we recall here:
\begin{equation}
  \eta=\frac{\rho\int v(x)\mathfrak K_eu(x)\ dx}{1-\rho\int v(x)\mathfrak K_e(2u(x)-\rho u\ast u(x))\ dx}
  .
  \label{etaf2}
\end{equation}

\point We now turn to\-~(\ref{eta_asym}).
First, note that, the estimate in (\ref{Kvu}) also holds for
\begin{equation}
  \int (\mathfrak{K}_ev) u\ast u\ dx   \leqslant  \|\mathfrak{K}_ev \|_2\|u\|_2  \leqslant \frac1{2^{\frac{11}{4}}\pi^{\frac32}}\|v\|_1^2  (2e)^{-\frac12} 
  \label{neglect_denom}
\end{equation}
which shows that
\begin{equation}
  \eta = \int v\mathfrak{K}_e(\rho u)\ dx\, +o(\rho)
  .
\end{equation}
By the resolvent identity, we rewrite
\begin{equation}
  \mathfrak K_e(\rho u) = \xi  -\mathfrak K_e(v \xi)
  \label{uresolventxi}
\end{equation}
with
\begin{equation}
  \xi :=  \mathfrak{Y}_e(\rho u)
\end{equation}
in which $\mathfrak Y_e$  is defined in \eqref{sim76}.

\subpoint We first prove that
\begin{equation}
  \xi(x)=\frac{\sqrt{2e}}{3\pi^2}+o(\sqrt e)
  \label{expandxi}
\end{equation}
uniformly in $x$.
We work in Fourier space: by \eqref{uhat},
\begin{equation}
  \widehat\xi(2\sqrt{e}k)=
  \frac1{4e}
  \left(\frac{k^2+1}{\sqrt{(k^2+1)^2-\frac\rho{2e}\widehat S(2\sqrt{e}k)}}-1\right)
  .
  \label{x1rescale}
\end{equation}
Since $S(x)\geqslant 0$, $|\widehat S(k)|\leqslant|\widehat S(0)|=\frac{2e}\rho$, and since $S$ is symmetric, $\widehat S$ is real, so
\begin{equation}
  |\widehat\xi(2\sqrt{e}k)|\leqslant
  \frac1{4e}
  \left(\frac{k^2+1}{\sqrt{(k^2+1)^2-1}}-\frac{k^2+1}{\sqrt{(k^2+1)^2+1}}\right)
  \label{dominatedxi1}
\end{equation}
which is integrable.
Next, note that $\frac\rho{2e}\widehat S(2\sqrt e k)\to1$ and
\begin{equation}\label{idint}
  \int\left(\frac{k^2+1}{\sqrt{(k^2+1)^2-1}}-1\right)\ \frac{dk}{8\pi^3}=\frac1{3\pi^2\sqrt2}
\end{equation}
which yields the leading order term in\-~(\ref{expandxi}).
Next, by\-~(\ref{x1rescale}) and\-~(\ref{idint}),
\begin{equation}
  \begin{array}{>\displaystyle r@{\ }>\displaystyle l}
    \xi(x)-\frac{\sqrt{2e}}{3\pi^2}
    =&
    \frac{\sqrt{e}}{4\pi^3} \int (e^{-i2\sqrt{e}kx}-1)
    \left(
      \frac{k^2+1}{\sqrt{(k^2+1)^2-\frac \rho{2e}\widehat S(2\sqrt{e}k)}}-1
    \right)\ dk
    \\[0.5cm]
    &+
    \frac{\sqrt{e}}{4\pi^3}\int
    \left(
      \frac{k^2+1}{\sqrt{(k^2+1)^2-\frac \rho{2e}\widehat S(2\sqrt{e}k)}}
      -\frac{k^2+1}{\sqrt{(k^2+1)^2-1}}
    \right)\ dk
    .
  \end{array}
\end{equation}
By\-~(\ref{dominatedxi1}), the first integrand is absolutely integrable, so
\begin{equation}
  \frac{\sqrt{e}}{4\pi^3} \int (e^{-i2\sqrt{e}kx}-1)
  \left(
    \frac{k^2+1}{\sqrt{(k^2+1)^2-\frac \rho{2e}\widehat S(2\sqrt{e}k)}}-1
  \right)\ dk
  =o(\sqrt e)
\end{equation}
uniformly in $x$.
Furthermore, since $\frac\rho{2e}\widehat S\leqslant 1$ and $1-(1+\epsilon)^{-\frac12}\leqslant\frac\epsilon 2$ for all $\epsilon\geqslant 0$,
\begin{equation}
  \left|
    \frac{k^2+1}{\sqrt{(k^2+1)^2-\frac \rho{2e}\widehat S(2\sqrt{e}k)}}
    -\frac{k^2+1}{\sqrt{(k^2+1)^2-1}}
  \right|
  \leqslant
  \frac{k^2+1}{((k^2+1)^2-1)^{\frac32}}\frac{1-\frac\rho{2e}\widehat S(2\sqrt{e}k)}2
\end{equation}
and since $\widehat S$ is the Fourier transform of $(1-u)v$ which is absolutely integrable, $\widehat S$ is uniformly continuous, so
\begin{equation}
  \frac{\sqrt{e}}{4\pi^3}\int dk\ 
  \left|
    \frac{k^2+1}{\sqrt{(k^2+1)^2-\frac \rho{2e}\widehat S(2\sqrt{e}k)}}
    -\frac{k^2+1}{\sqrt{(k^2+1)^2-1}}
  \right|
  =o(\sqrt{e})
  .
\end{equation}
This proves\-~(\ref{expandxi}).

\subpoint
By\-~(\ref{uresolventxi})
\begin{equation}
  \eta=\int v\xi\ dx-\int (\mathfrak K_e v)v\xi\ dx+o(\rho)
\end{equation}
and by\-~(\ref{expandxi}), since $v\mathfrak K_e v$ is integrable (which follows from Lemma\-~\ref{frakKprops}),
\begin{equation}
  \eta=\frac{\sqrt{2e}}{3\pi^2}\left(\int v\ dx-\int v\mathfrak K_e v\ dx\right)+o(\sqrt e)
  .
  \label{etaasym1}
\end{equation}
Furthermore,
\begin{equation}
  \lim_{e\to0}\mathfrak K_e
  =
  (-\Delta + v)^{-1}
  .
\end{equation}
Therefore, by dominated convergence, (we have $v\mathfrak K_ev \leqslant v(-\Delta+v)^{-1}v$ which is integrable)
\begin{equation}
  \int v(x)\mathfrak K_e v(x)\ dx
  \to
  \int v(x)\varphi(x)\ dx
\end{equation}
where $\varphi$ is the solution of the scattering equation $(-\Delta +v)\varphi= v$.
Furthermore, by \cite[Lemma 4.2]{CJL20},
\begin{equation}
  \int v(x)\varphi(x)\ dx
  =-4\pi a_0+\int v(x)\ dx
  .
  \label{intscat}
\end{equation}
Inserting this into\-~(\ref{etaasym1}), we find
\begin{equation}
  \eta=\frac{4\sqrt{2e}}{3\pi}a_0+o(\sqrt e)
  .
\end{equation}
We conclude the proof of\-~(\ref{eta_asym}) using the fact that $e=2\pi\rho a_0+o(\rho)$, which was proved in \cite[Theorem~1.4]{CJL20}.
\qed

\begin{remark}
If we knew that $2u - \rho u*u \geq 0$, we would have from Lemma~\ref{kl1} and \eqref{intu} that
\begin{equation}
\rho\int dx\ \mathfrak K_ev(x)(2u(x)-\rho u\ast u(x))  < \rho\int dx\ (2u(x)-\rho u\ast u(x)) =1\ ,
\label{ineq2uu}
\end{equation}
 and then we would know that  $\eta \geq0$.  We can at least prove the positivity of $\eta$  for small and large $\rho$.  By Lemmas~\ref{ul2b} and \ref{frakYbnd},
\begin{equation}
  \label{Kvu}
  \int (\mathfrak{K}_ev) u  dx \leqslant  \|\mathfrak{K}_ev \|_2\|u\|_2  \leqslant \frac1{2^{\frac{11}{4}}\pi^{\frac32}}\|v\|_1^2  (2e)^{-\frac12} 
\end{equation}
Therefore $\eta \geqslant 0$ for 
${\displaystyle \rho e^{-1/2} \leqslant \frac{2^{\frac{13}{4}} \pi^2}{\|v\|_1^2}}$.
\end{remark}

\section{The momentum distribution -- Proof of Theorem \ref{theo:tan}}\label{sec:momentum_distribution}
\indent
Equation\-~(\ref{frakM}) follows directly by differentiating\-~(\ref{simpleq_momentum}), in a computation that is very similar to the proof of\-~(\ref{etaf}) in section\-~\ref{sec:condensate_fraction}.
As was the case there, one first needs to prove that the solution $u_\mu$ exists and is differentiable, and, as before, we leave those details to the reader.
We now turn to the proof of\-~(\ref{M_asym}).
First of all, by\-~(\ref{neglect_denom}),
\begin{equation}
  \mathfrak M(k)=
  \rho\widehat u_0(k)(1+O(\rho))\int v(x)\mathfrak K_e\cos(k\cdot x)\ dx
\end{equation}
Proceeding as in section\-~\ref{sec:condensate_fraction}, we use the resolvent identity to rewrite
\begin{equation}
  \int v(x)\mathfrak K_e\cos(kx)\ dx
  =
  \int v(x)\mathfrak Y_e\cos(kx)\ dx
  -
  \int (v\mathfrak K_ev)(\mathfrak Y_e\cos(kx))\ dx
\end{equation}
in which $\mathfrak Y_e$  is defined in \eqref{sim76}, so
\begin{equation}
  \int v(x)\mathfrak K_e\cos(kx)\ dx
  =
  \frac{\widehat v(k)-\int e^{ikx}v\mathfrak K_e v\ dx}{k^2+4e(1-\rho \widehat u(k))}
  .
\end{equation}
Since $v\mathfrak K_e \leqslant v(-\Delta+v)^{-1}v$ which is integrable: by\-~(\ref{intscat})
\begin{equation}
  \int v(x)(-\Delta+v)^{-1}v(x)
  =
  -4\pi a_0+\int v(x)\ dx
\end{equation}
we have, by dominated convergence, in the limit $e\to0$ and $|k|\to0$,
\begin{equation}
  \int v(x)\mathfrak K_e\cos(kx)\ dx
  \sim
  \frac{4\pi a_0}{k^2+4e(1-\rho\widehat u(k))}
\end{equation}
so, as $\kappa\to\infty$,
\begin{equation}
  \int v(x)\mathfrak K_e\cos(kx)\ dx
  \sim\frac{\pi a_0}{e\kappa^2}
  .
\end{equation}
We conclude the proof of the theorem using $e\sim2\pi\rho a_0$\-~\cite[Theorem\-~1.4]{CJL20}.
\qed

\section{Bounds on $u$ and $\mathfrak{K}_e$ -- Proof of Lemma~\ref{ul2b} and Lemma~\ref{frakKprops}}\label{UB}

\begin{lemma}\label{uu1Lem}
Let $K_{e} := (-\Delta + v + 4e)^{-1}$ and
\begin{equation}\label{sim2}
u_1 := K_{e}v\ . 
\end{equation}
For all $p\geq 1$, 
\begin{equation}\label{sim5}
\|u_1\|_p \leqslant \|u\|_p \leqslant 2\|u_1\|_p\ .
\end{equation}
Furthermore, the operator $K_{e}$ has a positive kernel $K_e(x,y)$ satisfying 
\begin{equation}\label{bapos}
K_e(x,y) \leqslant  Y_{4e}(x-y)
\end{equation}
 where $Y_{4e}(x):=\frac{e^{-2\sqrt{e}|x|}}{4\pi|x|}$, so that $Y_{4e}(x-y)$ is the kernel of $(-\Delta + 4e)^{-1}$. 
\end{lemma}

\begin{proof}   Note that $K_e = \int_0^\infty dt e^{-4et} e^{t(\Delta - v)}$ and $e^{t(\Delta -v)}$ has a positive kernel by the Trotter product formula. Hence $K_e$ has a positive kernel $K_e(x,y)$. By the resolvent identity it then follows that
$K_e(x,y) \leqslant Y_{4e}(x-y)$, the kernel for $(-\Delta +4e)^{-1}$. 
By (\ref{simp}),
$u -u_1 = 2e\rho K_{e} u*u$, and since
$K_{e}(x,y)\geqslant 0$ is non-negative, so $ u_1\leqslant u$, and moreover, by \eqref{bapos},
\begin{equation}\label{sim3}
 u -u_1 \leqslant 2e\rho Y_{4e} *u*u\ , 
\end{equation}
Since $\|Y_{4e}\|_1 = \frac{1}{4e}$, it follows from Young's inequality and~(\ref{intu}) that, for all $p \geq 1$
\begin{equation}\label{sim3}
 \|u -u_1\|_p  \leqslant \frac12 \|u\|_p\ , 
\end{equation}
and hence
$ \|u\|_p \leqslant \|u-u_1\|_p + \|u_1\|_p \leqslant \frac12 \|u\|_p + \|u_1\|_p$.
\end{proof}

\begin{lemma}\label{highdens1}   Let $v\in L^2$. Then 
$$
\|u\|_2 \leqslant \frac{1}{2e}\|v\|_2\ .
$$
\end{lemma}

\begin{proof} By Lemma~\ref{uu1Lem},
$$\|u_1\|_{2}^2 \leqslant \frac{1}{(2\pi)^d} \int (k^2 + 4e)^{-2} \widehat{v}^2(k){\rm d}k \leqslant  \frac{1}{16e^2}\|v\|_2^2\ .$$
Now apply Lemma~\ref{uu1Lem} once more.
\end{proof}

We are now ready to prove Lemma~\ref{ul2b} which gives bounds on $\|u\|_p$ for $1 \leqslant p < 3$.

\begin{proof}[Proof of Lemma~\ref{ul2b}]
As in the previous proof,
$u_1 = K_{e}v  \leqslant Y_{4e}*v$.
Therefore, $\|u_1\|_p \leqslant \|Y_{4e}\|_p\|v\|_1$,
and
$$\|Y_{4e}\|_p^p   = (4\pi)^{1-p}(2p\sqrt{e})^{p-3} \Gamma(3-p)\  . $$
Then by Lemma \ref{uu1Lem}, $\|u\|_p \leqslant 2\|u_1\|_p$.   
The final statement is given by Lemma~\ref{highdens1}.
\end{proof}

We now turn to the proof of Lemma~\ref{frakKprops}.

\begin{lemma}\label{frakYbnd}  For  all non-negative $\psi\in L^1(\R^3)$,
${\displaystyle \|\mathfrak{K}_e\psi\|_2  \leqslant   \frac{1}{2\pi}(2e)^{-1/4}\|\psi\|_1}$.  For all real $\psi\in L^1(\R^3)$,  
${\displaystyle \|\mathfrak{K}_e\psi\|_2  \leqslant   \frac{1}{\pi}(2e)^{-1/4}\|\psi\|_1}$.
\end{lemma}

\begin{proof} 
Let $\mathfrak{Y}_e$ be defined as in \eqref{sim76}. Then the positive kernels of the operators $\mathfrak{K}_e$ and $\mathfrak{Y}_e$ are related by \eqref{KYcomp}, and hence
 for all positive $\psi$ 
$$ 0 \leqslant  \mathfrak{K}_e \psi \leqslant \mathfrak{Y}_e\psi \qquad{\rm and\ hence}\qquad \|\mathfrak{K}_e \psi\|_2 \leqslant   \|\mathfrak{Y}_e \psi\|_2 \ .$$ 
Then by \eqref{sim76}, for non-negative $\psi\in L^1$, $\mathfrak{Y}_e\psi\in L^2(\R^3)$ with 
\begin{equation}\label{sim30}\ 
\|\mathfrak{Y}_e\psi \|_2^2  \leqslant  \frac{\|\psi\|_1^2}{(2\pi)^3}  \int dk  [k^2 +4e (1- \rho\widehat u(k))]^{-2} \ .
\end{equation}

Recall from \eqref{uhat} that 
$$\rho \widehat{u}(k) = 1 + \frac{k^2}{4e} -  \sqrt{ \left(1 + \frac{k^2}{4e}\right)^2 -\frac{\rho}{2e}\widehat{S}(k)}\quad{\rm where}\quad \widehat{S}(k) = \int v(1-u)e^{-ikx} dx\ ,$$
and hence
$$(1- \rho \widehat{u}(k)) \geqslant    \sqrt{ \left(1 + \frac{k^2}{4e}\right)^2 -\frac{\rho}{2e}\widehat{S}(k)}\ .$$
By \eqref{simp} $\frac{\rho}{2e}S(k) \leqslant 1$ and hence 
${\displaystyle 
4e(1- \rho \widehat{u}(k))   \geqslant  \sqrt{k^4 + 8e k^2} \geqslant \sqrt{8e}|k|}$.
Therefore
$$
\int  [k^2 +4e (1- \rho\widehat u(k))]^{-2}dk \leqslant  \int dk [k^2 + \sqrt{8e}|k|]^{-2} = 4\pi \int_0^\infty dr \frac{r^2}{[r^2 + \sqrt{8e} r]^2}  =\frac{2\pi}{\sqrt{2e}}\ .
$$
For the final part, we apply the bound just proved to the positive and negative parts of $\psi$ separately. 
\end{proof}

\begin{remark}\label{frakYbndR} Let $\widetilde{\mathfrak{K}_e}$ be the operator defined in \eqref{sim81}.  It is easy to see that the same proof yields essentially  the same bound for this operator.
\end{remark}

\begin{lemma}\label{kl1} For $v\geqslant 0$, and $v\in L^1(\R^3)\cap L^2(\R^3)$, $0\leqslant \mathfrak{K}_e v(x) \leqslant 1$ for all $x$, and $\mathfrak K_e v(x)<1$ on a set of positive measure.
\end{lemma}

\begin{proof}  As we have already explained, $\mathfrak{K}_e$ preserves positivity and since $v \leqslant 0$, $\mathfrak{K}_e v(x) \geqslant 0$. Next,
for $\delta > 0$, define 
$$\mathfrak{K}_{e,\delta} := (-\Delta + v + 4e(I - C_{\rho u}) + \delta)^{-1} = \int_0^\infty dt e^{-t\delta} e^{\Delta - v -4e(I - C_{\rho u})}\ .$$
As $\delta$ decreases toward $0$, $\mathfrak{K}_{e,\delta}v(x)$ increases toward $\mathfrak{K}_{e,\delta}v(x)$.   To show that $\mathfrak{K}_{e}v(x) \leqslant 1$,  it suffices to show that
$\mathfrak{K}_{e,\delta}v(x) \leqslant 1$ for all $\delta > 0$.  

We next show that $\mathfrak{K}_{e,\delta}v$ belongs to the Sobolev space $H_{2}(\R^3)$, and hence is continuous.  This is a consequence of  the Kato-Rellich Theorem  \cite{RS75b}: 
since the operator $4e- 4e \rho C_{ u}$ is bounded and accretive on all $L^p$,  and since we are assuming $v\in L^{2}$, follows that the domain of $-\Delta  + \delta +v + 4e(I- C_{\rho u})$ is the same as the domain of $-\Delta  + \delta$, which is $H_2(\R^3)$, and moreover, $(-\Delta + v + 4e(I - C_{\rho u}) + \delta)^{-1} $ maps $L^2(\R^3)$ onto $H_2(\R^3)$. Since functions in $H_2(\R^3)$ are continuous, $\mathfrak{K}_{e,\delta}v$ is continuous.
Let $\psi :=    \mathfrak{K}_{e,\delta}v$, so that , by (\ref{bg6}),
$$\Delta \psi = v(\psi -1) + 4e(\psi - \rho u*\psi)  +\delta \psi\  .$$
Now, suppose $\|\psi\|_\infty > 1$. Then, by~(\ref{intu}), $\|\rho u*\psi\|_\infty < \|\psi\|_\infty$. Let $r:= (\|\psi\|_\infty - \|\rho u*\psi\|_\infty)/2$ and 
$$U := \{x\ :\ \psi > \max\{\|\psi\|_\infty - r,1\}\}\ ,$$ 
which is open.  Then evidently $\psi$ is subharmonic on $U$, and hence is maximized on the boundary, but this is impossible, so $U$ is empty, which is, again, impossible. Hence $\|\psi\|_\infty \leqslant 1$.   Finally by Lemma~\ref{frakYbnd}, since $v\in L^1(\R^3)$, 
$\mathfrak{K}_e v \in L^2(\R^3)$, the set on which this function exceed exceeds $1/2$ has finite measure. 
\end{proof}

\begin{lemma}\label{sym}  For all $\psi,\phi\in L^1(\R^3)\cap L^2(\R^3)$,
$$
 -\infty <   \int_{\R^3} \phi(x) ( \mathfrak{K}_e \psi)(x) \ dx=   \int_{\R^3} (\mathfrak{K}_e  \phi )(x)\  \psi(x)\ dx < \infty\ .
$$  
\end{lemma}

\begin{proof}  For all $\delta>0$, the operator $\mathfrak{K}_{e,\delta}$ defined in the proof of Lemma~\ref{kl1} is bounded and self adjoint so that   $\int_{\R^3} \phi(x) ( \mathfrak{K}_{e,\delta} \psi)(x)\ dx =   \int_{\R^3} (\mathfrak{K}_{e,\delta}  \phi )(x)\  \psi(x)\ dx$
By the monotonicity noted in the proof of Lemma~\ref{kl1}, the Lebesgue Dominated Convergence Theorem applied separately to the positive and negative parts shows that 
$\lim_{\delta\downarrow 0}\|\mathfrak{K}_{e,\delta}\phi - \mathfrak{K}_{e}\phi\|_2= 0$ and  $\lim_{\delta\downarrow 0}\|\mathfrak{K}_{e,\delta}\psi - \mathfrak{K}_{e}\psi\|_2= 0$.  The finiteness then follows from Lemma~\ref{frakYbnd}.
\end{proof}

\begin{lemma}\label{mfKcont} For any $\phi\in L^1(\R^3)$  function $e\mapsto \int_{\R^3} (\mathfrak{K}_e v)\phi dx $  is continuous, and consequently, if $e\mapsto \phi_e$ is continuous into $L^1(\R^3)$, then 
$e\mapsto \int_{\R^3} (\mathfrak{K}_e v)\phi_e {\rm d}x $ is continuous.
\end{lemma}

\begin{proof} Let $e_1,e_2> 0$, and let $\rho_j = \rho(e_j)$ and $u_j = u(\cdot,e_j)$,  $j=1,2$.  Then by the 
 resolvent identity,
\begin{equation}\label{sim78}
\mathfrak{K}_{e_2} v - \mathfrak{K}_{e_1}v  =  4\mathfrak{K}_{e_1} [  (e_1-e_2)  + e_2\rho_2C_{u_2}-e_1\rho_1 C_{u_1}   ] \mathfrak{K}_{e_2} v\ .
\end{equation}
Multiply by $\phi$ and integrate.
\begin{eqnarray}\label{sim78b}
\left|\int_{\R^3} \phi(\mathfrak{K}_{e_2} v) {\rm d}x - \int_{\R^3}\phi (\mathfrak{K}_{e_1} v) {\rm d}x \right|  &=& \left| \int_{\R^3}   (\mathfrak{K}_{e_1} \phi)[  4(e_1-e_2)  + 4e_2\rho_2C_{u_2}-e_1\rho_1 C_{u_1}   ] \mathfrak{K}_{e_2} v {\rm d}x \right|\nonumber\\
&\leqslant&\|\mathfrak{K}_{e_1} \phi\|_2 \|[  4(e_1-e_2)  + 4e_2\rho_2C_{u_2}-e_1\rho_1 C_{u_1}  ] \mathfrak{K}_{e_2} v\|_2\nonumber\\
&\leqslant& [  4|e_1-e_2|  + 4\|e_2\rho_2 u_2 -e_1\rho_1u_1\|_1  ]  \|\mathfrak{K}_{e_1} \phi\|_2  \|\mathfrak{K}_{e_2} v\|_2
\end{eqnarray}
The claim now follows from the bound in Lemma~\ref{frakYbnd} and the continuity of $e\mapsto e\rho(e)u(e,\cdot)$ in $L^1(\R^3)$. The final claim is now evident. 
\end{proof}

\begin{remark}\label{mfKcontR} Let $\widetilde{\mathfrak{K}_e}$ be the operator defined in \eqref{sim81}.  It is easy to see that the same proof yields the same statement for  $e\mapsto \int_{\R^3} (\widetilde{\mathfrak{K}_e} v)\phi {\rm d}x $ (keeping $\tilde e$ fixed).
The same also holds if we exchange $e$ and $\tilde{e}$.
\end{remark}

\begin{proof} [Proof of Lemma~\ref{frakKprops}] Lemma~\ref{frakYbnd} proves \eqref{frakKL2}.  Lemma~\ref{sym} proves \eqref{frakKprops2}.
Lemma~\ref{mfKcont} proves \eqref{frakKprops4}.
 Lemma~\ref{kl1} proves \eqref{frakKprops3}.
\end{proof}

\section{A variant of the Hardy-Littlewood-Sobolev inequality  -- Proof of Theorem~\ref{HLSvar0}}  

Recall that we have defined in \eqref{Gdef0}  the operator
\begin{equation}\label{Gdef}
\G_\beta f(x) =\int_{\R^d}|x-y|^{-\beta}f(y) dy\ .
\end{equation}
where $0 < \beta < d$.
Up to a constant multiple, depending only on $\beta$, $\G_\beta$ is the operator $(-\Delta)^{-(d-\beta)/2}$.
By the HLS inequality, for $1 < p < q < \infty$,  related by $\frac1q = \frac1p - \frac{d-\beta}{d}$,  there is a constant $C$ depending only on $d$, $\beta$ and $p$ such that
\begin{equation}\label{GHLS}
\|\G_\beta f\|_q  \leqslant C \|f\|_p 
\end{equation}
holds for all $f\in L^p$.  Evidently we must have $1 < p \leqslant \frac{d}{d-\beta}$ and hence $ d/\beta < q < \infty$\ .   To see that $q=d/\beta$ is unobtainable, note that if 
$f\geq 0$ has compact support,
$$\lim_{x\to \infty}|x|^\beta \int |x-y|^{-\beta}f(x){\rm d}x  = \int f dx \ .$$
Thus for $f\geq 0$, it is never the case that $\G f\in L^q$ for $q \leqslant d/\beta$ unless $f=0$. 

Our goal in this section is to prove a Theorem asserting that if $f$ is integrable with $\int_{\R^d} f dx = 0$, and if $f$ decays sufficiently rapidly at infinity, then $\|\G_\beta f\|_q$ will be bounded for certain $q \leqslant d/\beta$.  We introduce a norm that measures the decay of $f$ at infinity, and this involves Lorentz norms.   We briefly recall the relevant facts:

For a measurable set $A$, let $|A|$ denote the measure of $A$. 
For $1\leqslant p,q \leqslant \infty$, the Lorentz $p,q$ quasi-norm of a function $f$ is 
\begin{equation}\label{isias10}
\|f\|^*_{p,q} := p^{1/q}\left( \int_0^\infty  (\lambda |\{x\ :\  |f(x)| > \lambda\}|^{1/p})^q \frac{d\lambda}{\lambda}\right)^{1/q}\ .
\end{equation}
Then $L_{p,q}$ consists of the measurable functions $f$ such that $\|f\|_{p,q} < \infty$. For $p < \infty$, $\|f\|_{p,\infty}^* < \infty$ if and only if for some finite constant $C$, 
\begin{equation}\label{isias72}
|\{x\ :\  |f(x)| > \lambda\}| \leqslant C\lambda^{-p}\ .
\end{equation}
That is $L_{p,\infty}$ is weak $L^p$ and hence $L_{p,\infty} \subset L^p$.  
and in this case, 
$$
\|f\|_{p,\infty} \leqslant C^{1/p}
$$

Next consider $q=1$. 
\begin{equation}\label{isias70}
\|f\|^*_{p,1} = p \int_0^\infty  |\{x\ :\  |f(x)| > \lambda\}|^{1/p} d\lambda\ .
\end{equation}
By the definitions, Chebychev's inequality and the layer cake representation,
\begin{eqnarray*}
\|f\|_p^p &=& p \int_0^\infty  \lambda^{p-1}|\{x\ :\  |f(x)| > \lambda\}|d\lambda\\
 &=& p\int_0^\infty  \left(\lambda^{p-1} |\{x\ :\  |f(x)| > \lambda\}|^{(p-1)/p}\right)|\{x\ :\  |f(x)| > \lambda\}|^{1/p}d\lambda \\
 &\leqslant& p \|f\|_p^{p-1} \|f\|_{p,1}\ .
\end{eqnarray*}
Thus, $\|f\|_p \leqslant p \|f\|_{p,1}$ and hence $L_{p,1} \subset L^p$.  It would be natural to refer to $L_{p,1}$ as {\em strong $L^p$}, but this terminology is not standard.

 For $1 < p < \infty$ and $1\leqslant q \leqslant \infty$, $\|\cdot \|_{p,q}^*$ is equivalent to an actual norm,  $\|\cdot \|_{p,q}$
given by 
\begin{equation}\label{isias11}
\| f\|_{p,q} = \sup_g\left\{ \int |fg|dx \ : \ \|g\|^*_{p',q'} \leqslant 1\right\}\ .
\end{equation}

In particular, $\|f\|_{p,1}$ is bounded by a universal multiple of 
\begin{equation}\label{isias12}
 \int_0^\infty   |\{x\ :\  |f(x)| > \lambda\}|^{1/p} d\lambda\ .
\end{equation}

\begin{definition}  Let $f$  be a function such that $(1+|x|)^s f(x)\in L^1$, and such that $f\in L_{d/(d-\beta),1}$ ($L_{d/(d-\beta),1}$ is the Lorentz space with indices $d/(d-\beta),1$). 
Let $f_{\leqslant R}$ denote $f$ multiplied by the indicator function of the  closed ball of radius $R$, and let $f_{>R} := f - f_{\leqslant R}$.  
Given $s>d$, define
\begin{equation}\label{newnorm}
|\! |\! | f |\!|\!|_{\beta,s} =  \int_{\R^d} (1+|x|)^{s-d} |f(x)| dx +  \sup_{R>0} (1 +R)^{\beta+s - d}\|f_{>R}\|_{d/(d-\beta),1} 
\end{equation}
Define the space $\mathcal{L}_{\beta,s}$ to be the space of all measurable functions $f$ for which $|\! |\! | f |\!|\!|_{\beta,s} < \infty$. 
\end{definition} 

 We show below that if $f$ satisfies the bound 
\begin{equation}\label{isias1}
|f(x)| \leqslant M (1+ |x|)^{-r}\ ,
\end{equation}
then $f\in  \mathcal{L}_{\beta,s}$ for all $s < r$.   Taking $R= 0$ in \eqref{newnorm}, we see that $\|f\|_{d/(d-\beta),1} \leqslant |\! |\! | f |\!|\!|_{\beta,s}$, and since $L^{d/(d-\beta)}  \subset L_{d/(d-\beta),1}$, $L^{d/(d-\beta)} \subset   \mathcal{L}_{\beta,s}$.
Evidently, $L^1 \subset \mathcal{L}_{\beta,s}$. Thus, $\mathcal{L}_{\beta,s} \subset L^1 \cap  L^{d/(d-\beta)}$, and hence 
for all $f\in \mathcal{L}_{\beta,s}$ with $\beta$, $s$ as specified, $f\in L^p$ for all $1\leqslant p \leqslant d/(d-\beta)$; i.e., the whole HLS interval  including the endpoints.

\begin{theorem}\label{HLSvar}  Let $f\in  \mathcal{L}_{\beta,s}$ for some $d+1 \geq  s>d$  satisfying $\int_{\R^d} f(x)dx  =0$.  
 Then for all  $q \leqslant d/\beta$  such that 
 \begin{equation}\label{HLSV1}
q > \frac{d}{\beta+s -d}\ ,
\end{equation}
there is a constant $C$ depending only on $q$, $s$ and $\beta$ such that 
\begin{equation}\label{HLSV2}
\|\G_\beta f\|_{q} \leqslant C   |\! |\! | f |\!|\!|_{\beta,s} 
\end{equation}
Furthermore, for all $1 < p < \infty$, there is a constant $C$  depending only on $p$, $q$, $s$ and $\beta$ such that 
\begin{equation}\label{HLSV21}
\|\G_\beta f\|_{q} \leqslant C \left(\|f\|_{d/(d-\beta/p')}^{1-\theta}  |\! |\! | f |\!|\!|_{\beta,s}^\theta
\right)
\end{equation}
where
\begin{equation}\label{HLSV22}
\theta := \frac{dp-\beta q}{qp}  \frac{p'}{\beta + p'(s-d)}\ 
,\quad p':=\frac p{p-1}
\end{equation}
\end{theorem}

\begin{remark}  If we take $p\to \infty$ in \eqref{HLSV22}, we find the limiting value of $\theta$ is $\theta = \frac{d}{q(\beta+s-d)}$. 
Since  for all choices of $p$, $1 \leqslant d/(d-\beta/p') \leqslant d/(d-\beta)$ the remark after the definition of the norm  $ |\! |\! | \cdot  |\!|\!|_{\beta,s}$  provides
$$
\|\cdot \|_{d/(d-\beta/p')} \leqslant C  |\! |\! | \cdot  |\!|\!|_{\beta,s}
$$
and hence \eqref{HLSV2} follows from \eqref{HLSV21}.  However, it may be that $\|f \|_{d/(d-\beta/p')}$ is much smaller than  $|\! |\! | f  |\!|\!|_{\beta,s}$, as in our application, and then \eqref{HLSV21} gives better bounds. 
\end{remark}

\begin{lemma}  Let $f\in  \mathcal{L}_{\beta,s}$ for some $s>d$ satisfying $\int_{\R^d}f dx =0$.  Then for a universal constant $C$,
$$|\G_\beta( f_{\leqslant {|x|/2}})(x)| \leqslant  C |\! |\! | f |\!|\!|_{\beta,s}  |x|^{-(\beta+s-d)}\ .
$$
\end{lemma}

\begin{proof}  By the Fundamental Theorem of Calculus, for $|y|< |x|$,
\begin{eqnarray*}
|x+y|^{-\beta}  - |x|^{-\beta} &=& \int_0^1 \frac{{\rm d}}{{\rm d}t} (|x +ty|^2)^{-\beta/2}{\rm d}t\\
&=& -\beta\int_0^1 \left[(x\cdot y +t|y|^2) (|x+ty|^2)^{-(\beta+2)/2} \right]{\rm d}t\\
\end{eqnarray*}
Then for $|x|\geqslant 2R$, with $R$ to be chosen below, and $|y|\leqslant R$,
\begin{equation}\label{isias21}
||x+y|^{-\beta}  - |x|^{-\beta}| \leqslant \beta 2^{\beta+2} |x|^{-\beta-2}(|x||y| + |y|^2)\ .
\end{equation}
Therefore,
$$\left| \int |x-y|^{-\beta} f_{\leqslant R}(y)dy  - |x|^{-\beta} \int  f_{\leqslant R}(y)dy\right| \leqslant \beta2^{\beta+2} |x|^{-\beta-2}\int (|x||y| + |y|^2)| f_{\leqslant R}(y)|dy\ ,
$$
and then since
$$
 \int  f_{\leqslant R}(y)dy = -  \int  f_{>R}(y)dy\ ,
 $$
\begin{equation}\label{isias31}
|\G_\beta f_{\leqslant R}(x)| \leqslant  |x|^{-\beta} \int  |f_{>R}|(y)dy + \beta2^{\beta+2} |x|^{-\beta-2}\int (|x||y| + |y|^2)| f_{\leqslant R}(y)|dy\ .
\end{equation}

 Next,
\begin{equation}\label{isias1}
\int  | f_{>R}(y)|dy   \leqslant R^{d-s} \int (1 + |y|)^{s-d}  | f(y)|dy \leqslant R^{d-s} |\! |\! | f |\!|\!|_{\beta,s}\ ,
\end{equation}
and then since  $0\leqslant s-d \leqslant 1$, $|y| \leqslant |y|^{1 +d -s}  (1+|y|)^{s-d}  \leqslant R^{1+d-s}(1+|y|)^{s-d}\ ,$ 
\begin{equation}\label{isias2}
\int  |y| | f_{\leqslant R}(y)|dy   \leqslant R^{d+1-s} \int (1 + |y|)^{s-d}  | f(y)|dy \leqslant R^{1+d-s} |\! |\! | f |\!|\!|_{\beta,s}\ ,
\end{equation}
and similarly if $s-d \leqslant 2$,
\begin{equation}\label{isias3}
\int  |y|^2 | f_{\leqslant R}(y)|dy   \leqslant R^{d+2-s} \int (1 + |y|)^{s-d}  | f(y)|dy \leqslant R^{2+d-s} |\! |\! | f |\!|\!|_{\beta,s}\ .
\end{equation}
Using \eqref{isias1} through \eqref{isias3} in \eqref{isias31} yields
$$
\left|\int |x-y|^{-\beta} f_{\leqslant R}(y)dy\right|  \leqslant \left( \frac{1}{|x|^\beta  R^{s-d}} +  \frac{\beta2^{\beta+2} }{|x|^{\beta+1} R^{s-d-1}}  + \frac{\beta 2^{\beta+2} }{|x|^{\beta +2} R^{s-d-2}} \right) |\! |\! | f |\!|\!|_{\beta,s}\ .
$$
Taking $R = |x|/2$ we have the desired bound.  
\end{proof}

In the next lemma, we use O'Neil's extension of H\"older's inequality to the Lorentz seminorms \cite{On63}, recalled below. A special case says that for a universal constant $M$,
$$\int|fg|dx \leqslant M\|g\|_{p,\infty}\|f\|_{p',1}\ .$$
 $L_{p,\infty}$ is weak $L^p$ and hence for  $p=\tfrac{d}{\beta}$, $\||x|^{-\beta}\|_{d/\beta,\infty} < \infty$.

\begin{lemma}  Let $f\in \mathcal{L}_{\beta,s} \cap L^3$.  Then for a universal constant $C$,
$$|\G_\beta( f_{> {|x|/2}})(x)| \leqslant  C|\! |\! | f |\!|\!|_{\beta,s}(1+  |x|)^{-(\beta +s -d)}\ .
$$
\end{lemma}

\begin{proof}  Fix any $R>0$. Since $|x|^{-\beta}\in L_{d/\beta,\infty}$,  O'Neil's inequality gives
$$\||x|^{-\beta }*f_{> R}\|_\infty   \leqslant C  \||x|^{-\beta}\|_{d/\beta,\infty} \|f_{>R}\|_{d/(d-\beta),1}$$
Thus, for another universal constant C, we have from the definition of the norm $|\! |\! | \cdot |\!|\!|_{\beta,s}$ that 
$$
|\G_\beta (f_{> R})(x)| \leqslant C|\! |\! | f |\!|\!|_{\beta,s} (1+R)^{-(\beta +s -d)}\ .
$$
Again choosing $R = |x|/2$ yields the result. 
\end{proof}

\begin{proof}[Proof of Theorem~\ref{HLSvar}]  By the previous two lemmas, 
$$|\G_\beta f(x)| \leqslant C|\! |\! | f |\!|\!|_{\beta,s}  |x|^{-(\beta+s-d)} $$
for some universal $C$. Pick $q$ such that 
\begin{equation}\label{isias41}
1 \leqslant  q \leqslant d/\beta\ . 
\end{equation}  For any $R>0$ we decompose  $\G_\beta f = (\G_\beta f)_{\leqslant R} + (\G_\beta f)_{> R}$, and will estimate the $L^q$ norm of $(\G_\beta f)_{\leqslant R}$ using the HLS inequality as follows.  Pick $p>1$  so that $\|\G_\beta f\|_{pd/\beta}$ can be bounded using the HLS inequality. Then by H\"older's inequality,
$$
\|(\G_\beta f)_{\leqslant R}\|_q^q \leqslant  \|\G_\beta f\|_{pd/\beta }^q (|B_d|R^d)^{1- \beta q/pd}
$$
where $|B_d|$ is the volume of the $d$-dimensional unit ball.
By the HLS inequality, since $1 < \tfrac{d}{d - \beta/p'} < \tfrac{d}{d-\beta}$, 
$$\|\G_\beta f\|_{pd/\beta } \leqslant C \| f\|_{d/(d - \beta/p')}\ .$$
Then, if $C$ is a universal constant that may change from line to line,
\begin{eqnarray*}
\int_{\R^d}|\G_\beta f|^q dx &=&   \int_{\R^d}|(\G_\beta f)_{\leqslant R}|^q dx +  \int_{\R^d}|(\G_\beta f)_{>R}|^q dx\\
&\leqslant&  C \| f\|_{d/(d - \beta/p')}^q  (|B_d|R^d)^{1- \beta q/pd} \\
 &+& d|B_d|  (C|\! |\! | f |\!|\!|_{\beta,s} )^q \int_R^\infty  t^{d-1- q(\beta+s-d)} dt\\
&\leqslant&  C\left( \|f\|_{d/(d-\beta/p')}^q R^{d- \beta q/p}  + |\! |\! | f |\!|\!|_{\beta,s}^q  R^{d -q(\beta +s-d)}\right)\ ,
\end{eqnarray*}
where the last displayed integral is convergent on account of \eqref{HLSV1}.   In the final line, the first exponent on $R$ is positive on account of \eqref{isias41}, and the second is negative on account of   \eqref{HLSV1}.
Now choose 
$$R = \left(  \frac{|\! |\! | f |\!|\!|_{\beta,s} }   { \|f\|_{d/(d-\beta/p') }}\right)^{\frac{p'}{\beta +p'(s-d)}}\ .
$$

\end{proof}

The next lemma provides a simple estimate on the $\mathcal{L}_{\beta,s}$ norm that we shall apply in the next section. 

\begin{lemma}\label{lnorm1}
Suppose that for  some constant $M$, $|f(x)| \leqslant M (1+ |x|)^{-r}$.  Then for all $s<r$, there is a constant $C$ depending only on $s$ such that 
\begin{equation}\label{decay15}
|\! |\! | f |\!|\!|_{\beta,s}  \leqslant CM \ .
\end{equation}
\end{lemma}

\begin{proof}
Note that $\{x\ :\  |f(x)| > \lambda\} \subset \{ x\ :\ |x| \leqslant (M/\lambda)^{1/r}\}$,
and hence
\begin{equation}\label{decay2}
|\{x\ :\  |f_{>R}(x)| > \lambda\}| \leqslant \begin{cases} |B_d|  \left(\frac{M}{\lambda}\right)^{d/r} & \lambda \leqslant M(1+R)^{-r}\\ 0 &  \lambda > M(1+R)^{-r} \end{cases}\ .
\end{equation}
One then computes that for $p'>\frac dr$,
$$
\|f_{> R}\|_{p',1}^* \leqslant  M^{d/rp'}  p' |B_d|^{1/p'}  \int_0^{M(1+R)^{-r}} \lambda^{-d/rp'}  d\lambda = M p' |B_d|^{1/p'} \left(\frac{1}{1- d/rp'}\right)(1+R)^{ d/p'-r}
$$
We take $p=d/\beta $, so that $p' =d/(d-\beta)> d/r$,
and hence we have 
$$
\|f_{> R}\|_{d/(d-\beta),1}^* \leqslant   M p' |B_d|^{1/p'} \left(\frac{1}{1- d/rp'}\right)(1+R)^{ d- \beta -r}  \leqslant CM (1+R)^{d-\beta -s}\ .
$$
Next,
$$
\int_{\R^d} (1+|x|)^{s-d} |f(x)| dx \leqslant  M \int_{\R^d} (1+|x|)^{s-r-d}  dx  \leqslant CM\ .
$$
\end{proof}

\section{Bounds on $u'$ -- Proof of Theorem~\ref{Lpu'}   }\label{sec:uprime}  Recall from \eqref{upform} that
\begin{equation}\label{bamb34}
u' =  \mathfrak{K}_e \varphi \qquad{\rm where }\quad \varphi = 2\left( 1 + \rho'\frac{e}{\rho}\right)\rho u*u -4u\ ,
\end{equation}
and since $u\in L^p$ for all $p$,  $\varphi\in L^p$ for all $p$.   
We need bounds on $\|u'\|_q$ for small values of $q$, close to $1$, and to obtain these we shall need bounds on $\|u'\|_q$ for large values of $q$. These are relatively easy to obtain using what we know about the operators $ \mathfrak{Y}_e$ and $ \mathfrak{K}_e$.

The operator $ \mathfrak{Y}_e$ stands in the same relation to $\mathfrak{K}_e= (-\Delta + v + 4e(I - C_{\rho u}))^{-1}$ as does $  G_e := (-\Delta  + 4e)^{-1}$ to $K_e  := (-\Delta + v + 4e)^{-1}$. In particular, all four operators have positive kernels and 
we have that for all $x,y$
\begin{equation}\label{bamb1}
K_e(x,y) \leqslant G_e(x,y) \qquad{\rm and}\qquad  
\mathfrak{K}_e(x,y) \leqslant \mathfrak{Y}_e(x,y)\ .
\end{equation}

By the resolvent identity
$\mathfrak{K}_e  = K_e  + 4e K_e  C_{\rho u} \mathfrak{K}_e$.  Therefore, by \eqref{bamb1}, for all $f\geq 0$,
$$
\mathfrak{K}_e f  \leqslant G_e f +  4e G_e  C_{\rho u} \mathfrak{Y}_e f\ .
$$
The kernel of $G_e$, $G_e(x,y) = Y_{4e}(x-y)$ where the function $Y_{4e}$ is defined in Lemma~\ref{uu1Lem}.  Since $Y_{4e}$ is a probability density,  it follows from Young's Inequality that $4eG_e$ is a contraction on $L^p$ for all $p$.  Likewise, since 
$\rho u$ is a probability density, it follows in the same way that 
$C_{\rho u}$ is a contraction on $L^p$ for all $p$. Hence for any non-negative function $f$, and any $q\geq 1$, 
$$
\|\mathfrak{K}_e f \|_q  \leqslant \| G_e f \|_q+  \|4e G_e  C_{\rho u} \mathfrak{Y}_e f\|_q  \leqslant  \| G_e f \|_q+  \| \mathfrak{Y}_e f\|_q
$$
Applying this to the positive and negative parts of $\varphi$ separately, yields
\begin{equation}\label{bamb32}
\|u'\|_q \leqslant \|\mathfrak{K}_e \varphi_+ \|_q + \|\mathfrak{K}_e \varphi_- \|_q  \leqslant 2 \| G_e \varphi \|_q+  2\| \mathfrak{Y}_e \varphi\|_q
\end{equation}
We now estimate these terms to prove:

\begin{lemma}\label{up2} The is a constant $C$ independent of $e$,  such that for all $3/2 < q < \infty$,
$$\|u'\|_q \leqslant C e^{-1/2  -3/2q}\ .$$
\end{lemma}

\begin{proof}Since $G_e \varphi = Y_{4e}* \varphi$, Young's Inequality yields   
$$
\|G_e \varphi\|_q \leqslant \|Y_{4e}\|_p\|\varphi\|_r
$$
where $1+1/q = 1/p + 1/r$.   Since $Y_{4e}\in L^p$ for $1 \leqslant p < 3$ with $\|Y_{4e}\|_p \leqslant Ce^{-(3-p)/2p}$. (In what follows, the constant denoted by $C$ will change from line to line.)  Also,  $\|\varphi\|_p  \leqslant C\|u\|_p$ for all $p$. 
by the definition of $\varphi$ and the bound on $\rho'$ provided by Theorem~\ref{Mon}, 
Taking $q=\infty$, and $p=r= 2$,  we  have from Lemma~\ref{ul2b}, 
$\| G_e\varphi \|_\infty \leqslant Ce^{-1/2}$.
Taking $q= 3/2$, and $p=r = 6/5$,
$\| G_e\varphi \|_{3/2} \leqslant Ce^{-3/2}$.
It then follows that for $3/2 \leqslant q \leqslant \infty$,
\begin{equation}\label{bamb31}
\| G_e\varphi \|_q \leqslant Ce^{-1/2  -3/2q}
\end{equation}

 The operator $\mathfrak{Y}_e$ is also a convolution operator, but somewhat more complicated. It has a useful factorization that we now describe.
 
 Note that the Fourier transform renders $\mathfrak{Y}_e$ as the operation of multiplication by $[k^2 + 4e(1- \rho \widehat{u}(k))]^{-1}$.
 Recall that 
$$\rho \widehat{u}(k) = 1 + \frac{k^2}{4e} -  \sqrt{ \left(1 + \frac{k^2}{4e}\right)^2 -\frac{\rho}{2e}\widehat{S}(k)}\ ,$$
where 
$\widehat{S}(k) = \int v(1-u)e^{-ikx} dx$,  $\frac{\rho}{2e}S(k) \leqslant 1$ and hence
\begin{equation}\label{sim40}
[k^2 + 4e(1- \rho \widehat{u}(k))]^{-1}
= \frac{1+\widehat H(k)}{k(k^2 + 8e)^{1/2}}
\end{equation}
with
$$\widehat{H}(k) := \left(1 + \frac{16e^2(1-\frac{\rho}{2e}\widehat{S}(k))}{ k^4 + 8ek^2 } \right)^{-1/2}  -1.$$
Since $\lim_{k\to\infty}\widehat{S}(k) = 0$, $\widehat{H}$ is integrable. 
 
It is more work to see that its inverse Fourier transform, $H(x)$ is also integrable, but we show this below.  It turns out that $H(x)$ is not non-negative. Had this been the case we would have that
$\|H\|_1 =\widehat{H}(0)$.  To compute this, we expand 
\begin{equation}
  \frac{\rho}{2e}\widehat{S}(k) = 1 +  \beta  k^2 + O(k^4)  .
\end{equation}
Therefore,
$$ 
\int H(x)dx = \widehat{H}(0) = (1+2e\beta)^{1/2} -1 \leqslant e\beta\ .
$$
The following lemma shows that $\|H\|_1$ is not quite so small for small $e$, but is indeed still small:

\begin{lemma}\label{Hlem}  There is a constant $C$ depending  independent of $e$ such that 
$$| H(x)| \leqslant e^{1/2} \frac{C}{(1+|x|^2)^2}\ ,$$
and in particular, for a different $C$ still independent of $e$,  $\| H\|_1 \leqslant e^{1/2}C $. 
\end{lemma}
\begin{proof}[Proof of Lemma~\ref{Hlem}]
Recall that
$$\widehat{H}(k) := \left(1 + G(k) \right)^{-1/2}  -1 \qquad{\rm where}\qquad G(k) :=  \frac{16e^2(1-\frac{\rho}{2e}\widehat{S}(k))}{ |k|^4 + 8e|k|^2}\ .$$
The proof is very similar to that of Theorem~\ref{theo:pointwise}, except simpler:  One shows that  $\widehat{H}$ and $\Delta^2 \widehat{H}$ are integrable, with $\|\widehat{H}\|_1 +  \|\Delta^2 \widehat{H}\|_1 \leqslant Ce^{1/2}$. The claim now follows from the Riemann-Lebesgue lemma. 
\end{proof}

We can now explicitly specify the factorization of  $\mathfrak{Y}_{e}$ mentioned above:  Let $\mathcal{R}$ be the Riesz potential operator acting by $\widehat{\mathcal R}\psi = |k|^{-1}\widehat{\psi}(k)$. 
Let $\mathcal{B}$ be the Bessel potential operator acting by $\widehat{\mathcal B}\psi = (|k|^2 + 8e)^{-1/2}\widehat{\psi}(k)$. 
Let $\mathcal{H}$ be the  operator acting by $\widehat{\mathcal H}\psi = \widehat{H}(k)\widehat{\psi}(k)$.     Then by \eqref{sim40},
\begin{equation}\label{bamb35}
\mathfrak{Y}_{e} = \mathcal{H}\mathcal{B}\mathcal{R}   +   \mathcal{B}\mathcal{R}  \ .
\end{equation} 
Since $\mathcal{H}$ and $\mathcal{B}$ are convolution operators with integrable kernels, they are bounded on $L^p$ for all $p$, with  $B$, the kernel of $\mathcal{B}$ satisfying
 $\|B\|_1 = (8e)^{-1/2}$.  Hence  there is a constant $C$ independent of $e$ such that 
 \begin{equation}\label{monk1}
 \|\mathfrak{Y}_{e}\varphi\|_q  \leqslant Ce^{-1/2}\|\mathcal{R}\varphi\|_q\ .
 \end{equation}
Since
$$\mathcal{R}\varphi(x) = \frac{1}{2\pi^2}\int |x-y|^{-2} \varphi(y) dy =:  \frac{1}{2\pi^2} \mathcal{G}_2 \varphi(x)\ ,$$
 we can apply the Hardy-Littlewood-Sobolev (HLS)  Inequality to estimate $\|\mathcal{R}\varphi\|_q$:
For $3/2 < q < \infty$,  related by $\frac1q = \frac1p - \frac{1}{3}$,  there is a constant $C$ depending only on  $q$ such that
\begin{equation}\label{GHLS}
\|\mathcal{R} f\|_q  \leqslant C \|f\|_{3q/(q+3)} 
\end{equation}
holds for all $f\in L^p$.

Going back to \eqref{monk1}, we obtain  $\|\mathfrak{Y}_{e}\varphi\|_q \leqslant Ce^{-1/2}\|\varphi\|_{3q/(3+q)}$, and then by the definition of $\varphi$ and the bound on $\rho'$ provided by Theorem~\ref{Mon}, with $C$,  changing from line to line,
\begin{equation}\label{comboi5}
\|\mathfrak{Y}_{e}\varphi\|_q \leqslant Ce^{-1/2}\|u\|_{3q/(3+q)}\ .
\end{equation}
By Lemma~\ref{ul2b}, $\|u\|_p \leqslant C e^{-(3-p)/2p}$, and hence we  obtain
\begin{equation}\label{comboi5}
\|\mathfrak{Y}_{e}\varphi\|_q \leqslant Ce^{-3/2q}\ .
\end{equation}
Using this and \eqref{bamb31} in \eqref{bamb32} we see that for small $e$, \eqref{bamb31} is the dominant term.
\end{proof}

To estimate $\|u'\|_p$ for $p \leqslant 3/2$, we will use  Theorem~\ref{HLSvar}, but now we need a different expression for $u'$.  By the resolvent identity,
$$\mathfrak{K}_e = \mathfrak{Y}_e - \mathfrak{Y}_ev \mathfrak{K}_e\ ,$$
and then by \eqref{bamb34},
\begin{equation}\label{bamb67}
u' = \mathfrak{K}_e\varphi  = \mathfrak{Y}_e\varphi  - \mathfrak{Y}_ev \mathfrak{K}_e\varphi  = \mathfrak{Y}_e(\varphi - vu')\ .
\end{equation}
We therefore define $\psi := \varphi - vu'$. By Lemma~\ref{up2}, $u'\in L^2$, and since we assume that $v\in L^2$,  $vu'$ is integrable. Thus, $\psi$ is integrable.
Furthermore,
\begin{equation}\label{intphi}
\int \psi dx  = \frac{2e}{\rho^2}\rho' - \frac{2}{\rho} -\int vu' dx
=\frac{2e}{\rho^2}\rho' - \frac{2}{\rho} + \frac{{\rm d}}{{\rm d}e}\frac{2e}{\rho} = 0\ .\\
\end{equation}

We next show that $\varphi$ inherits a bound of the form $|\varphi(x)| \leqslant Ce^{-3/2}(1+|x|^4)^{-1}$ from $u$. 

\begin{lemma} Let $f$ be any non-negative  function satisfying $f(x) \leqslant \frac{C}{1 +|x|^4}$ and $0 <   \int f(x) dx =:\sigma^{-1}$.  
$$\sigma f*f(x) \leqslant  \frac{C}{1 +|x/2|^4}$$
for the same constant $C$.  
\end{lemma}

\begin{proof}  $$\int_{|y|> |x|/2} \sigma f(x-y)f(y) dy \leqslant \int_{|y|> |x|/2} \sigma f(x-y) \frac{C}{1 +|y|^4} dy   \leqslant  \int_{|y|> |x|/2} \sigma f(x-y) \frac{C}{1 +|x/2|^4} dy  \leqslant \frac{C}{1 +|x/2|^4}\ ,$$
and
$$
\int_{|y|<  |x|/2} \sigma f(x-y)f(y) dy \leqslant \int_{|y|< |x|/2}\frac{C}{1 +|x-y|^4}\sigma f(y) dy \leqslant \int_{|y|< |x|/2}\frac{C}{1 +|x/2|^4}\sigma f(y) dy = \frac{C}{1 +|x/2|^4}\ .
$$
\end{proof}

It now follows that with $e_\star$ defined as in  Theorem~\ref{Mon}, on account of the bound on $\rho'$ proved there, and on account of Theorem~\ref{theo:pointwise} that there is a constant independent of $e$ such that for all $e\leqslant e_\star$,
\begin{equation}\label{varphipointwise}
|\varphi(x)| \leqslant Ce^{-3/2}(1 + |x|^4)^{-1}\ .
\end{equation}

Now Lemma~\ref{lnorm1} provides an estimate on  $|\! |\! | \varphi   |\!|\!|_{2,s} $ for all $s < 4$.   We then need a bound on  $|\! |\! | vu'   |\!|\!|_{2,s} $, and for this we shall use the estimate on $\|u'\|_q$ for large $q$ that we have just proved.

\begin{lemma}\label{lnorm2}  Let $v$ be such that   $(1+|x|^3)v(x) \in L^8(\R^3)$ in addition to our usual hypothesis that $(1+|x|^4)v(x)\in L^1(\R^3)\cap L^2(\R^3)$.   Then for all $3<s<4$ there is a constant $C$ such that for all $e \leqslant e_\star$ such that
$$
|\! |\! | vu' |\!|\!|_{2,s}  \leqslant C e^{-5/4}  \ .
$$
\end{lemma}

\begin{proof}
We first estimate $\|v_{> R}u'\|_{3,1}$.  For small $\lambda$, we use
$$|\{ x \ :\ |v_{> R}u'(x)| > \lambda\}| \leqslant\frac{\|v_{>R}u'\|_1}{\lambda}\leqslant\frac{\|v_{>R}\|_2\|u'\|_2}{\lambda}\ ,$$ and hence for any $L>0$,
$$\int_0^L  |\{ x \ :\ |v_{> R}u'(x)| > \lambda\}|^{1/3}d\lambda \leqslant \frac32  \|v_{>R}\|_2^{1/3}\|u'\|_2^{1/3} L^{2/3}\ .$$

For large $\lambda$, we use
$$|\{ x \ :\ |v_{> R}u'(x)| > \lambda\}| \leqslant\frac{\|v_{>R}u'\|_4^4}{\lambda^4}\leqslant\frac{\|v_{>R}\|_8^4\|u'\|_8^4}{\lambda^4}\ ,$$ and hence for any $L>0$,
$$\int_L^\infty  |\{ x \ :\ |v_{> R}u'(x)| > \lambda\}|^{1/3}d\lambda \leqslant 3  \|v_{>R}\|_8^{4/3}\|u'\|_8^{4/3} L^{-1/3}\ .$$
Optimizing in $L$, we find
\begin{eqnarray}
\int_0^\infty  |\{ x \ :\ |v_{> R}u'(x)| > \lambda\}|^{1/3}d\lambda  &\leqslant& \frac92\left( \|v_{>R}\|_2^{1/3}\|u'\|_2^{1/3}  \right)^{1/3}\left( \|v_{>R}\|_8^{4/3}\|u'\|_8^{4/3} \right)^{2/3}\\
&=& \frac92\|u'\|_2^{1/9} \|u'\|_8^{8/3} \|v_{>R}\|_2^{1/9} \|v_{>R}\|_8^{8/9}\\
&\leqslant& \frac92\|u'\|_2^{1/9} \|u'\|_8^{8/3} \|v\|_2^{1/9} \|v\|_8^{8/9} \ .\label{smallR}
\end{eqnarray}
By Lemma~\ref{up2}, $\|u'\|_2^{1/9} \|u'\|_8^{8/9} \leqslant Ce^{-3/4}$.
Also,  for all $p>0$,
$\|v_{>R}\|_8^8 \leqslant   R^{-p8}\||x|^pv\|_8^8$, so that $\|v_{>R}\|_8^{8/9} \leqslant R^{-p8/9}\||x|^pv\|_8^{8/9}$.  Likewise, $\|v_{>R}\|_2^{1/9}\leqslant  R^{-4/9} \| |x|^4 v\|_2^{1/9}$
Using\-~(\ref{smallR}) when $R<1$ and choosing $p=3$ when $R\geqslant 1$, we see that,
$$
 \sup_{R>0} (1 +R)^{2+s - 3}\|v_{>R}u'\|_{3,1}  \leqslant C  e^{-3/4}\||x|^3v\|_8^{8/9}  \| |x|^4 v\|_2^{1/9}\ .
$$
Finally, using Lemma~\ref{up2}
$$
\int_{\R^3} (1+ |x|)^{s-d} |vu'| dx \leqslant  \left(\int_{\R^3} (1+|x|)^2 v^2(x) d x\right)^{1/2}\|u'\|_2  \leqslant  Ce^{-5/4}  \left(\int_{\R^3} (1+|x|)^2 v^2(x) d x\right)^{1/2}\ .
$$
For $e \leqslant e_\star$, this is the dominant power of $e$. 
\end{proof}

\begin{proof}[Proof of Theorem~\ref{Lpu'}]  By Lemma~\ref{up2}, it only remains to get a bound on $\|u'\|_p$ for $1 < p< 3/2$, and for this  we make use of  \eqref{bamb67}.
By \eqref{bamb35} and \eqref{monk1}, for all $q> 1$,
$$\|u'\|_q \leqslant Ce^{-1/2}\|\mathcal{G}_2\psi\|_q\ $$
where $\psi = \varphi - vu'$.

We check that $\psi$   satisfies  assumptions of  Theorem~\ref{HLSvar}. 
First of all, by~(\ref{intphi}) and the  discussion just above it, $\varphi\in L^1(\mathbb R^3)$ and $\int \varphi=0$.    Next, by the triangle inequality,
$$ |\! |\! | \varphi  |\!|\!|_{2,s} \leqslant |\! |\! | 2\rho u*u  + 2e\rho' u*u  -4u   |\!|\!|_{2,s}  +  |\! |\! | vu'   |\!|\!|_{2,s}\ .$$
By Theorem~\ref{theo:pointwise},
there is a constant $C$ such that 
$$
|2\rho u*u(x)  + 2e\rho' u*u(x)  -4u(x)| \leqslant C\rho^{-1}e^{-1/2}(1+|x|)^{-4}\ .$$  
It now follows from Lemma~\ref{lnorm1} that for all $3 < s < 4$, 
$$
|\! |\! | 2\rho u*u  + 2e\rho' u*u  -4u   |\!|\!|_{2,s} \leqslant  C\rho^{-1}e^{-1/2}\ .
$$
Next it follows from Lemma~\ref{lnorm2} that 
$$
 |\! |\! | vu'   |\!|\!|_{2,s}  \leqslant C\rho^{-1}e^{-1/4}.
$$
Therefore,  for any $e_0>0$, there is a constant $C$ such that for all $e \leqslant e_0$, 
$$
 |\! |\! | \psi  |\!|\!|_{2,s}  \leqslant C\rho^{-1}e^{-1/2}\ .
$$

We now apply \eqref{HLSV21} of 
Theorem~\ref{HLSvar}: in the limit $s\to 4$ and $p\to \infty$, we would have 

$$\|\mathcal{R}\psi\|_q \leqslant \|\psi\|_3^{5/8} |\! |\! | \psi  |\!|\!|_{2,s}^{3/8}\ .$$
We have that $|\! |\! | \psi  |\!|\!|_{2,s} \leqslant C e^{-3/2}$, and for $e\leqslant e_\star$, the dominant contribution to $\|\psi\|_3$ comes from $\|vu'\|_3 \leqslant\|v\|_p\|u'\|_{3p/(p-3)}$ which holds for any $p>3$. By Lemma~\ref{up2}, using the assumption that $v\in L^8(\mathbb R^3)$, $\|vu'\|_3 \leqslant Ce^{-1+3/2p}$ for all $3<p\leqslant 8$.  This would yield a bound proportional to $e^{-19/16+15/16p}$.  We can come arbitrarily close to this, so, chosing $p<5$, we certainly have 
$$\|\mathcal{R}\psi\|_q \leqslant  C e^{-1}\ ,$$
and then the result follows from \eqref{monk1} with $\psi$ in place of $\varphi$. 
\end{proof}

\section{Explicit solution --  Proof of Theorem \ref{theo:explicit} }\label{sec:explicit}

Given an  integrable  function $u(x)$ satisfying $0 \leqslant u(x)\leqslant 1$,  we seek to find $\rho$, $e$ and a non-negative potential $v(x)$ such that $u(x)$ then solves \eqref{simp}.  Since we must have $\rho^{-1} = \int u dx$, the choice of $u$ fixes $\rho$. Then pick any $e>0$, and if $u$ solves \eqref{simp} for this $\rho$, $e$, and some potential $v(x)$, we must have 
\begin{equation}\label{vdef}
  v=\frac{-\Delta u+2e(2u-\rho u\ast u)}{(1-u)}\ .
\end{equation}
The remaining question then is whether this potential $v$ is non-negative and integrable. 
We look for a solution in the form of a Cauchy kernel
\begin{equation}
  u(x)=\frac c{(1+b^2x^2)^2}
  \label{u}
\end{equation}
since this has the expected decay rate and since it is easy to compute $u*u$.

By direct computations, we find that
\begin{equation}
  \rho=\frac1{\int dx\ u(x)}=\frac{b^3}{c\pi^2}
  \label{rho}
  ,\quad
  \widehat u(k)=\frac{\pi^2 c}{b^3}e^{-\frac{|k|}b}
  ,\quad
  u\ast u=\frac{2\pi^2c^2}{b^3(4+b^2 x^2)^2}
\end{equation}
so
\begin{equation}
  2u-\rho u\ast u=
  \frac{6c(5+2b^2x^2)}{(1+b^2x^2)^2(4+b^2x^2)^2}
  ,\quad
  \Delta u=\frac{12cb^2(x^2b^2-1)}{(1+b^2x^2)^4}
  .
  \label{2usu}
\end{equation}
Therefore,
\begin{equation}
    -\Delta u+2e(2u-\rho u\ast u)
    =
 12c
      \frac{ 
      x^6b^6(2e-b^2)
      +b^4x^4(9e-7b^2)
      +4b^2x^2(3e-2b^2)
      +(5e+16b^2)
   }{(1+b^2x^2)^4(4+b^2x^2)^2}
    .
\end{equation}

that is,
\begin{equation} \label{S}
  v(x)=
12c  \frac{x^6b^6(2e-b^2) +b^4x^4(9e-7b^2) +4b^2x^2(3e-2b^2) +(5e+16b^2)}{(1+b^2x^2)^2(4+b^2x^2)^2((1+b^2x^2)^2-c)}
  .
\end{equation}
The denominator is non-negative for $0 \leqslant c \leqslant 1$, and the leading power is $|x|^{12}$.   The leading power in the numerator is $|x|^6$,  
unless $2e = b^2$,  and the coefficient is non-negative if and only if $2e \geqslant  b^2$.  If  $2e =  b^2$, then the leading term is $|x|^4$, but with a negative coefficient.  So we must have at least $9e \geqslant  7 b^2$, all of the coefficients in the numerator are non-negative ,and hence the numerator is non-negative.
This provides the condition \eqref{cond_explicit}
\begin{equation}
  \frac e{b^2}\geqslant
  \frac79  \quad{\rm and}\qquad c \leqslant 1
  ,
  \end{equation}
  and we see that when this condition is satisfied, $v(x)$ decays at infinity like $|x|^{-6}$, so that $(1+|x|^2) v(x)$  is integrable, but with  $\int |x|^4v(x) dx = \infty$.
\qed

We conclude this section with some remarks.

\point
For large $|x|$,
\begin{equation}
  u\sim\frac c{b^4x^4}
  .
  \label{asymu}
\end{equation}
Theorem\-~\ref{theo:pointwise} states that if $x^4v$ were integrable, then
\begin{equation}
  u\sim\frac{\sqrt{2+\beta}}{2\pi^2\rho\sqrt ex^4}
  \label{asymupred}
\end{equation}
where by \eqref{betadef}
\begin{equation}
  \beta=\rho\int dx\ x^2S(x)
    .
\end{equation}
From\-~(\ref{u}) and\-~(\ref{S}) we find that
\begin{equation}
  \beta
  =
  \frac{6(2e-b^2)}{b^2}
  .
\end{equation}
Therefore, by\-~(\ref{rho})
\begin{equation}
  \frac{\sqrt{2+\beta}}{2\pi^2\rho\sqrt ex^4}
  =
  \frac{c}{b^4x^4}\sqrt{3-\frac{b^2}e}
  .
\end{equation}
This agrees with\-~(\ref{asymu}) if and only if
$2e=b^2$, but then the potential is negative for large $|x|$.
This is not a contradiction since our potentials $v(x)$ never satisfy $\int |x|^4 v dx < \infty$  when they are non-negative so that Theorem~\ref{theo:pointwise} does not apply.
However, this does imply that there are solutions, such as the one constructed above, in which $x^2v$ is integrable, $x^4v$ is not, and\-~(\ref{asymupred}) does not hold.

\point By\-~(\ref{2usu}), we see that
\begin{equation}
  2u-\rho u\ast u\geqslant 0
  .
\end{equation}
Recalling the discussion surrounding\-~(\ref{gb}) and\-~(\ref{ineq2uu}), the monotonicity of $e(\rho)$ and the fact that $\eta\geqslant 0$  would follow directly from this inequality.
The fact that it holds for the explicit solution is further evidence that this inequality holds in general.

\vskip20pt

\noindent{\bf Acknowledgements}:
{\it
We are very grateful to Markus Holzmann for many enlightening discussions on the physics of the Bose gas and for sharing is detailed numerical results on the Bose gas.
U.S.~National Science Foundation grants DMS-1764254 (E.A.C.),  DMS-1802170 (I.J.)   are gratefully acknowledged.
}


\bibliographystyle{amsalpha}

\bibliography{bibliography}

\end{document}